\documentclass{amsart}

\RequirePackage[OT1]{fontenc}
\RequirePackage{amsthm,amsmath}
\RequirePackage[numbers]{natbib}
\RequirePackage[colorlinks,citecolor=blue,urlcolor=blue]{hyperref}
\usepackage{mathtools}
\DeclarePairedDelimiter\ceil{\lceil}{\rceil}
\DeclarePairedDelimiter\floor{\lfloor}{\rfloor}

\usepackage{graphicx}
\usepackage{latexsym,color}
\usepackage{graphicx}
\usepackage{amssymb}
 \usepackage{amsfonts}
 \usepackage{amsmath}
\usepackage{amsthm}
\usepackage{float}

\newtheorem{thm}{Theorem}[section]
 
 \newtheorem{lem}{Lemma}[section]
 \newtheorem{rem}{Remark}[section]

\newtheorem{prop}{Proposition}[section]
\newtheorem{asm}{Assumption}[section]

\newtheorem{exm}{Example}[section]

\numberwithin{equation}{section}

\begin{document}

\title{Continuity of Utility Maximization under Weak Convergence}

 \author{Erhan Bayraktar \address{Department of Mathematics, University of Michigan} \email{erhan@umich.edu}} \thanks{E. Bayraktar is supported in part by  the National
 Science Foundation under grant DMS-1613170 and in part by the Susan M. Smith Professorship.}
 \author{Yan Dolinsky  \address{
  Department of Statistics, Hebrew University} \email{yan.dolinsky@mail.huji.ac.il}} \thanks{Y. Dolinsky is supported in part by the Israeli Science Foundation under Grant
  160/17.}
\author{Jia Guo \address{Department of Mathematics, University of Michigan}\email{guojia@umich.edu}}

\date{}


\maketitle \markboth{}{}
\renewcommand{\theequation}{\arabic{section}.\arabic{equation}}
\pagenumbering{arabic}

\begin{abstract}
In this paper we find tight sufficient conditions for the continuity of the
value of the utility maximization problem from terminal wealth with respect to
the convergence in distribution of the underlying processes.
We also establish a weak convergence result for the terminal
wealths of the optimal portfolios.
Finally, we apply our results to the computation of
the minimal expected shortfall (shortfall risk) in the Heston model by building an appropriate lattice
approximation.

\end{abstract}

\tableofcontents
\section{Introduction}\label{sec:1}
A basic problem of mathematical finance is the problem of an economic agent, who
invests in a financial market so as to maximize the expected utility of his terminal
wealth.
The problem of utility maximization problem is going back
to the seminal work by R. Merton \cite{Merton:69, Merton:71}
and continuing
e.g. in \cite{CH:89, KLSX:91, KS:99, CSW:01, CSW:18, HIM:05, RS:05, BF:08}.

This paper deals with the following question:
Given a utility function
and a sequence of financial markets with underlying assets $(S^{(n)})_{n\in\mathbb N}$ that are converging weakly
to $S$, under which conditions
do the values of the utility maximization problems (from terminal wealth) converge to the corresponding value
for the model given by $S$?

Our paper is motivated by the following economic applications.
First, in a general (incomplete)
financial market the problem of utility maximization
does not admit an explicit solution and so numerical
schemes come naturally into the picture.
Another important motivation is the practical limitations of calibration. Namely, we want to understand whether
the utility maximization problem is stable under small misspecifications of the law of the asset prices.

To the best of our knowledge, the continuity under weak convergence
was studied only in a complete market setup (see \cite{H:91,P:03, Reichlin:16}).
In this work we consider this convergence question for general incomplete market models and continuous
(as a function of the terminal wealth) random utility functions.

We divide the proof of our main result, namely Theorem~\ref{main}, into two main steps identifying when we have lower and upper-semi-continuity respectively.
We show that for the lower semi--continuity to hold, it is sufficient (in addition to some technical assumptions) that
the approximating sequence $(S^{(n)})_{n\in\mathbb N}$ has vanishing jump activity.
The formal condition is given in Assumption \ref{asm5.1}.
The main idea is to prove that an admissible integral of the form $\int\gamma dS$ can be approximated in the weak sense
by admissible integrals
of the form $\int \gamma^{(n)}d S^{(n)}$, $n\in\mathbb N$.
The assumption on the jump activity is essential for the admissibility of the approximating sequence.
We demonstrate the necessity of this assumption with an example; see Section~\ref{sec:necas3}.
We would like to emphasize that the concavity of the utility function is not necessary in this step.

The second step, namely, the upper semi--continuity is more delicate.
We prove that if the utility function is concave and the state price densities in the limit model
can be approximated by state price densities in the approximating sequence (see Assumption \ref{asm4.1})
then upper semi--continuity holds. The proof relies on the optional decomposition theorem.
In Sections~\ref{sec:necofas4f}--\ref{sec:necofas4s} we discuss the necessity of our assumptions.
Example \ref{exm4.2} in Section~\ref{sec:necofas4f} is surprising and quite interesting in its own right.
In this example we construct a sequence of complete market models (binomial models) which
converge weakly to an incomplete market model (a stochastic volatility model).

In addition to the convergence of the values, we prove a weak convergence for the terminal wealths of the optimal portfolios; see Theorem~\ref{thm2.2}.
An open question is whether
there is a convergence of the optimal trading strategies, i.e. of the integrands.
In a complete market setup convergence of the optimal trading strategies
were obtained in \cite{H:91,P:03}.
The proof was based on an explicit characterization of the optimal
trading strategies.
In the incomplete market setup we do not have explicit formulas for the optimal portfolios. Hence,
the problem is much more complicated and requires additional machinery (see Remark \ref{rem:integrand}).

We apply our continuity results in order to construct an approximating sequence verifying all our assumptions
for the Heston model in Section~\ref{sec:7}.
Our method is based on recombining trinomial trees and so,
for technical reasons we truncate the model in such a
way that the volatility is bounded. The novelty of our construction is that the
approximating sequence lies on a grid and satisfies the assumptions required for the
continuity of the value of the utility maximization problem from terminal wealth.
The grid structure enables efficient numerical computations for stochastic control
problems via dynamic programming.

Our last contribution, which is the subject of Section~\ref{sec:6}, is the implementation
of the constructed approximating models for the numerical computations in the Heston model.
For the shortfall risk measure we show that the truncation error can be controlled, see Lemma~\ref{error}, so our result applies to the non-truncated Heston model.
It is well known (see \cite{CPT:99,FS:99,DN:18,N:18}) that in the Heston model the super--replication price
is prohibitively high and lead to buy--and--hold strategies.
Namely, the cheapest way to super--hedge a European call option
is to buy one stock at the initial time and keep that position till maturity.
That is why the computation of shortfall risk is important. This cannot be done analytically and so numerical schemes come into picture.

A closely related topic to the one studied in the present paper is the
stability of the utility maximization problem under market parameters and the investor preferences.
Since the work \cite{JN:04} which dealt with complete markets, large progress
was made in the study of the stability of the utility maximization problem  in incomplete markets
(see, for instance, \cite{LZ:07,L:09,KZ:11,LY:12,BK:13,BS:18,LMZ:18,MS:18}).
The main difference from our setup is that in these papers the stochastic
base is fixed while in our setup each financial model is defined on
its own probability space. As a result, while the above cited papers deal with the stability
of the models with respect to small perturbations, we are able to obtain numerical
approximations using discrete models.

The rest of the paper is organized as follows. In the next section we introduce
the setup and formulate the main results.
In Section \ref{sec:nec} we discuss Assumptions \ref{asm5.1},\ref{asm4.1},\ref{asm4.2} and demonstrate their necessity.
In Section~\ref{sec:5} we prove the lower semi--continuity. In Section~\ref{sec:4}
we prove the upper semi--continuity. In Section \ref{sec:4WEAK} we establish
Theorem~\ref{thm2.2}.
Section~\ref{sec:7} is devoted to the construction
of an approximating sequence for the Heston model.
In Section~\ref{sec:6} we provide a detailed numerical
analysis for shortfall risk minimization.

\section{Preliminaries and Main Results}\label{sec:2}\setcounter{equation}{0}

We consider a model of a security
market which consists of $d$ risky assets which we denote by
$S=(S^{(1)}_t,...,S^{(d)}_t)_{0\leq t\leq T}$, where $T<\infty$ is the time horizon.
We assume that the investor has a bank account that,
for simplicity, bears no interest. The process $S$ is assumed to be a continuous semi--martingale on a filtered probability
space $(\Omega,\mathcal F,(\mathcal F^S_t)_{0\leq t\leq T},\mathbb P)$ where
the filtration
$(\mathcal F^S_t)_{0\leq t\leq T}$ is the usual filtration generated by $S$. Namely, the filtration
$\{\mathcal F^S_t\}_{t=0}^T$ is the minimal filtration which is complete, right continuous and satisfies
$\mathcal F_t\supset\sigma\{S_{u}: u\leq t\}$.
Without a loss of generality we take $\mathcal F:=\mathcal F^S_T$.

 A (self--financing) portfolio $\pi$ is defined as a pair $\pi=(x,\gamma)$ where the constant
$x$ is the initial value of the portfolio and $\gamma=(\gamma^{(i)})_{1\leq i\leq d}$
is a predictable $S$--integrable process
specifying the amount of each asset held in the portfolio.
The corresponding portfolio value process is given by
$$V^\pi_t:=x+\int_{0}^t\gamma_{u} dS_{u}, \ \ t\in [0,T].$$

Observe that the continuity of $S$ implies that the wealth process $\{V^\pi_t\}_{t=0}^T$
is continuous as well.
We say that a  trading strategy $\pi$ is admissible if
$V^\pi_t\geq 0$, $\forall t\geq 0.$
For any $x>0$ we denote by $\mathcal A(x)$
the set of all admissible trading strategies.

Denote by $\mathcal M(S)$ the set of all equivalent (to $\mathbb P$) local
 martingale measures. We assume that $\mathcal M(S)\neq\emptyset$.
 This condition is intimately related to the absence of arbitrage opportunities
on the security market. See \cite{DelbSch:94} for a precise statement and references.

Next, we introduce our utility maximization problem.
Consider a continuous function $U:(0,\infty) \times\mathbb D([0,T];\mathbb R^d)\rightarrow\mathbb R$.
As usual, $\mathbb D([0,T];\mathbb R^d)$ denotes the space
of all RCLL (right continuous with left limits) functions
$f:[0,T]\rightarrow\mathbb R^d$ equipped with the Skorokhod topology (for details see \cite{B:68}).
\begin{asm}\label{asm2.1}
${}$\\
(i) For any $s\in \mathbb D([0,T];\mathbb R^d)$ the function $U(\cdot,s)$ is non--decreasing.\\
(ii) For any $x>0$ we have
$\mathbb E_{\mathbb P}[U(x,S)]>-\infty$ .
\end{asm}
We extend $U$ to $\mathbb R_{+}\times\mathbb D([0,T];\mathbb R^d)$ by
$U(0,s):=\lim_{v\downarrow 0} U(v,s)$. In view of Assumption \ref{asm2.1}(i) the limit exists (might be $-\infty$).

For a given initial capital $x>0$ consider the optimization problem
$$
u(x):=\sup_{\pi\in\mathcal A(x)}\mathbb E_{\mathbb P}[U(V^\pi_T,S)],
$$
where we set $-\infty+\infty=-\infty$. Namely, for a random variable $X$ which satisfies $\mathbb E_{\mathbb P}[\max(-X,0)]=\infty$ we set
$\mathbb E_{\mathbb P}[X]:=-\infty$.

Let us notice that Assumption \ref{asm2.1}(ii) implies $u(x)>-\infty$.

\begin{asm}\label{asm2.1+}
The function $u:(0,\infty)\rightarrow \mathbb R\cup \{\infty\}$ is continuous.
Namely, for any $x>0$ we have
$u(x)=\lim_{y\rightarrow x} u(y)$
where a priori the joint value can be equal to $\infty$.
\end{asm}
Next, for any $n$,
let $S^{(n)}=(S^{n,1}_t,...,S^{n,d}_t)_{0\leq t\leq T}$
be a RCLL semi--martingale defined on some filtered probability space
$(\Omega_n,\mathcal F^{(n)},(\mathcal F^{(n)}_t)_{0\leq t\leq T},\mathbb P_n)$ where the filtration
$(\mathcal F^{(n)}_t)_{0\leq t\leq T}$ satisfies the usual assumptions (right continuity and completeness).
For the $n$--th model we define $\mathcal A_n(x)$
as the set of all pairs $\pi_n=(x,\gamma^{(n)})$ such that
$\gamma^{(n)}$ is a predictable $S^{(n)}$--integrable process
and the resulting portfolio value process
$$V^{\pi_n}_t:=x+\int_{0}^t \gamma^{(n)}_{u} dS^{(n)}_{u}\geq 0, \ \ t\in [0,T],$$
is non-negative.
Set, $$u_n(x):=\sup_{\pi_n\in\mathcal A_n(x)}\mathbb E_{\mathbb P_n}[U(V^{\pi_n}_T,S^{(n)})].$$

We assume the
weak convergence $S^{(n)}\Rightarrow S$ on the space $\mathbb D([0,T];\mathbb R^d)$ equipped with the Skorokhod topology
(for details see Chapter 4 in  \cite{B:68}).
Moreover, we assume the
following uniform integrability assumptions.
\begin{asm}\label{asm2.2}
${}$\\
(i)  For any $x>0$
the family of random variables $\{U^{-}(x,S^{(n)})\}_{n\in\mathbb N}$ is uniformly integrable
where $U^{-}:=\max(-U,0)$. \\
(ii) For any $x>0$ the family of random variables $\{U^{+}(V^{\pi_n}_T,S^{(n)})\}_{n\in\mathbb N, \pi_n\in\mathcal A_n(x)}$ is uniformly integrable, where
$U^{+}:=\max(U,0)$.
\end{asm}
\begin{rem}
The verification of Assumption \ref{asm2.1+}
and Assumption \ref{asm2.2}(ii) can be a difficult task.
In Section \ref{sec:AE} we provide quite general and easily verifiable conditions which are sufficient
for the above assumptions to hold true.
\end{rem}
Due to the admissibility requirements we will need the following assumption
which bounds the uncertainty of the jump activity. This assumption will be discussed in details in Section
\ref{sec:necas3}.
\begin{asm}\label{asm5.1}
For any $n\in\mathbb N$
consider the non-decreasing RCLL process given by $D^{(n)}_t:=\sup_{0\leq u\leq t} |S^{(n)}_{u}-S^{(n)}_{u-}|$, $t\in [0,T]$
where $|\cdot|$ denotes the Euclidean norm in $\mathbb R^d$.
For any $n$, there exists an adapted (to $(\mathcal F^{(n)}_t)_{0\leq t\leq T}$)
left continuous process $\{J^{(n)}_t\}_{t=0}^T$, $n\in\mathbb N$
such that
$\inf_{0\leq t\leq T}\left(J^{(n)}_t-D^{(n)}_t\right)\geq 0$ a.s.
and $J^{(n)}_T\rightarrow 0$ in probability.
\end{asm}
\begin{rem}
Let us notice that Assumption \ref{asm5.1}
and the weak convergence $S^{(n)}\Rightarrow S$ imply that $S$ is a continuous process. Indeed,
from Assumption \ref{asm5.1} we have
$\sup_{0\leq u\leq T} |S^{(n)}_{u}-S^{(n)}_{u-}|\rightarrow 0$. This together with Theorem 13.4 in \cite{B:68} yields that
$S$ is continuous.
\end{rem}
Now, we ready to formulate our first result (lower semi--continuity) which will be proved in Section \ref{sec:5}.
\begin{prop}\label{prop5.1}
Under Assumptions \ref{asm2.1}--\ref{asm2.1+}, Assumption \ref{asm2.2}(i)
and Assumption \ref{asm5.1} we have
$$u(x)\leq \lim\inf_{n\rightarrow\infty} u_n(x), \ \ \forall  x>0.$$
\end{prop}
Next, we treat upper semi--continuity.
\begin{asm}\label{asm4.1}
Recall the set $\mathcal M(S)$ of all equivalent local martingale measures.
Denote by $\mathcal M(S^{(n)})$, $n\in\mathbb N$ the set of all equivalent local martingale measures for the $n$--th model.
For any $\mathbb Q\in\mathcal M(S)$ there exists a sequence of probability measures $\mathbb Q_n\in\mathcal M(S^{(n)})$, $n\in\mathbb N$ such that
under $\mathbb P_n$ the joint distribution of
$\left(\{S^{(n)}_t\}_{t=0}^T,\frac{d\mathbb Q_n}{d\mathbb P_n}\right)$ on the space $\mathbb D([0,T];\mathbb R^d)\times \mathbb R$ converges to
the joint distribution of
$\left(\{S_t\}_{t=0}^T,\frac{d\mathbb Q}{d\mathbb P}\right)$ under $\mathbb P$. We denote this relation by
\begin{equation}\label{4.1}
\left(\left(S^{(n)},\frac{d\mathbb Q_n}{d\mathbb P_n}\right);\mathbb P_n\right)\Rightarrow
\left(\left(S,\frac{d\mathbb Q}{d\mathbb P}\right);\mathbb P\right).
\end{equation}
\end{asm}
\begin{rem}\label{AE:4}
The verification of Assumption \ref{asm4.1} requires a comfortable
representation of the corresponding local martingale measures. This is the case for
tree based approximations of diffusion processes. In Section \ref{sec:verification} we illustrate the verification of Assumption \ref{asm4.1}
for tree based approximations of the Heston model.

We do notice that in order to verify Assumption \ref{asm4.1}
it is sufficient to establish (\ref{4.1}) for a dense subset of
$\left\{\frac{d\mathbb Q}{d\mathbb P} :
\mathbb Q\in\mathcal M(S)\right\}$.
This simplification will be used in Section \ref{sec:verification}.
\end{rem}
\begin{rem}
In the complete market setup, Assumption \ref{asm4.1} is equivalent
to the joint convergence
of the underlying assets and of the Radon--Nikodym derivatives with respect to the unique risk neutral probability measure.
In this case (see \cite{P:03}) Assumption \ref{asm4.1} guarantees
the convergence of both the values and the optimal strategies.
\end{rem}

\begin{asm}\label{asm4.2}
For any $s\in \mathbb D([0,T];\mathbb R^d)$, the function $U(\cdot,s)$ is concave.
\end{asm}
Assumption \ref{asm4.2} says that the investor can not gain from additional randomization.

The following upper semi--continuity result will be proved in Section \ref{sec:4}.
\begin{prop}\label{prop4.2}
Under Assumption \ref{asm2.1}(i), Assumption \ref{asm2.2}(ii) and Assumptions \ref{asm4.1},\ref{asm4.2} we have
$$u(x)\geq \lim\sup_{n\rightarrow\infty} u_n(x), \ \ \forall  x>0.$$
\end{prop}
We now combine the statements of the above propositions and state them as the main theorem of our paper:
\begin{thm}\label{main}
Under Assumptions \ref{asm2.1}--\ref{asm2.2},\ref{asm5.1},\ref{asm4.1},\ref{asm4.2} we have
\begin{equation}\label{2.result}
u(x)=\lim_{n\rightarrow\infty} u_n(x), \ \ \forall x>0.
\end{equation}
\end{thm}
\begin{proof}
Follows from Proposition \ref{prop5.1} and Proposition \ref{prop4.2}.
\end{proof}
\begin{rem}
Observe that in view of Assumption \ref{asm2.2} we have
$$-\infty<\lim\inf_{n\rightarrow\infty} u_n(x)\leq \lim\sup_{n\rightarrow\infty} u_n(x)<\infty, \ \ \forall x>0.$$
We conclude that the joint value in (\ref{2.result}) is finite.
\end{rem}
Next, we establish the weak convergence for the optimal terminal wealths.
\begin{thm}\label{thm2.2}
Assume that Assumptions \ref{asm2.1}--\ref{asm2.2},\ref{asm5.1},\ref{asm4.1},\ref{asm4.2} hold true.
Moreover, assume that for any $s\in \mathbb D([0,T];\mathbb R^d)$ the function $U(\cdot,s)$ is strictly concave.
Let $x>0$
and $\hat\pi_n\in\mathcal{A}_n(x)$, $n\in\mathbb N$ be a sequence of asymptotically optimal portfolios, namely
\begin{equation}\label{2.asp}
\lim_{n\rightarrow\infty} \left(u_n(x)-\mathbb E_{\mathbb P_n}[U(V^{\hat\pi_n}_T,S^{(n)})]\right)=0.
\end{equation}
Then
$$\left(S^{(n)},V^{\hat\pi_n}_T\right)\Rightarrow \left(S,V^{\hat \pi}_T\right),$$
where $\hat \pi\in\mathcal A(x)$ is the unique portfolio that satisfies
$u(x)=\mathbb E_{\mathbb P}[U(V^{\hat\pi}_T,S)]$.
\end{thm}
The proof of Theorem \ref{thm2.2} will be given in Section \ref{sec:4WEAK}.
\begin{rem}
It is well known (see Theorem 2.2 in \cite{KS:99})
that for a utility function which is strictly concave there exists a unique optimizer.
Although in \cite{KS:99} the authors do not consider a random utility, their argument can be without much effort extended to our setup.
\end{rem}
\begin{rem}\label{rem:integrand}
A natural question is whether
Theorem \ref{thm2.2}
can be applied for establishing a convergence result for optimal trading strategies (the integrands).
It seems that this problem is closely related to the
robustness of martingale representations which was studied in
\cite{JMP:00}.

The main result in \cite{JMP:00} (Theorem A) deals with case where the underlying processes are martingales with respect to the same filtration, this is not satisfied in our setup.
Section 4 in \cite{JMP:00} deals with the weak convergence setup,
however the obtained results are limited to some particular cases
for which there are explicit representation for the integrands.
\end{rem}

\subsection{On the verification of Assumption \ref{asm2.1+}
and Assumption \ref{asm2.2}(ii)}\label{sec:AE}
${}$\\
The following result provides a simple and quite general condition which implies Assumption \ref{asm2.1+}.
\begin{lem}\label{AE:1}
Assume that Assumption \ref{asm2.1} holds true and there exist continuous functions
$m_1,m_2:[0,1)\rightarrow\mathbb{R}_{+}$ with $m_1(0)=m_2(0)=0$ (modulus of continuity) and
a non-negative random variable $\zeta\in L^1(\Omega,\mathcal F,\mathbb P)$ such that for any $\lambda\in (0,1)$ and
$v>0$
$$
U((1-\lambda) v,S)\geq (1-m_1(\lambda)) U(v,S)-m_2(\lambda) \zeta.
$$
Then Assumption \ref{asm2.1+} holds true.
\end{lem}
\begin{proof}
In view of the fact that $u$ is a non-decreasing function (follows from Assumption \ref{asm2.1}(i))
it sufficient to prove that for any $x>0$
\begin{equation*}
\lim_{\alpha\downarrow 0} u((1-\alpha)x)\geq \lim_{\alpha\downarrow 0} u((1+\alpha)x).
\end{equation*}
For any $\beta, y>0$
the map $(y,\{\gamma_t\}_{t=0}^T)\rightarrow (\beta y, \{\beta\gamma_t\}_{t=0}^T)$ is a bijection between $\mathcal A(y)$ and
$\mathcal A(\beta y)$. Thus,
\[
\begin{split}
&\lim_{\alpha\downarrow 0}u((1-\alpha)x) \\
&\geq\lim_{\alpha\downarrow 0} \left(\left(1-m_1\left(1-\frac{1-\alpha}{1+\alpha}\right)\right) u((1+\alpha)x)-
m_2\left(1-\frac{1-\alpha}{1+\alpha}\right)\mathbb E_{\mathbb P}[\zeta]\right)\\
&=\lim_{\alpha\downarrow 0} u((1+\alpha)x).
\end{split}
\]
\end{proof}
\begin{rem}\label{AE:new}
We notice that the power and the log utility satisfy the assumptions of Lemma \ref{AE:1}.
On the other hand for these utility functions
Assumption \ref{asm2.1+} is straightforward.

A ``real" application of Lemma \ref{AE:1} is the case which corresponds to the utility function given by
(\ref{5.utility}). In this case,
if $v\geq \frac{S_T}{1-\lambda}$ then
$U((1-\lambda) v,S)=U(v,S)=0$. If $v<\frac{S_T}{1-\lambda}$ then
$|U((1-\lambda) v,S)-U(v,S)|\leq \lambda v\leq \frac{\lambda}{1-\lambda} S_T$.
Thus, for
$m_1(\lambda):= 0$, $m_2(\lambda):=\frac{\lambda}{1-\lambda}$ and $\zeta:=S_T$ the assumptions of Lemma \ref{AE:1} hold true
(provided that $\mathbb E_{\mathbb P}[S_T]<\infty$).
\end{rem}
Next, we treat Assumption \ref{asm2.2}(ii).
\begin{lem}\label{AE:2}
Suppose there exist constants $C>0$, $0<\gamma<1$ and $q>\frac{1}{1-\gamma}$ which satisfy the following. \\
(I) For all $(v,s)\in (0,\infty) \times\mathbb D([0,T];\mathbb R^d)$,
\begin{equation}\label{bound}
U(v,s)\leq C (1+v^{\gamma}).
\end{equation}
(II) For any $n\in\mathbb N$ there exists
a local martingale measure ${\mathbb Q}_n\in\mathcal M(S^{(n)})$ such that
\begin{equation}\label{boundmeasure}
\sup_{n\in\mathbb N}\mathbb E_{{{\mathbb Q}}_n}\left[\left(\frac{d{\mathbb P}_n}{d{\mathbb Q}_n}\right)^{q}\right]<\infty.
\end{equation}
Then Assumption \ref{asm2.2}(ii) holds true.
\end{lem}
\begin{proof}
Let $p=\frac{q}{q-1}$.  Clearly $\frac{1}{p}>\gamma$. Thus in view of (\ref{bound}), in order to prove that Assumption \ref{asm2.2}(ii)
holds true,
it suffices to show that for any $x>0$
$$\sup_{n\in\mathbb N}\sup_{\pi_n\in\mathcal A_n(x)}\mathbb E_{\mathbb P_n}[(V^{\pi_n}_T)^{1/p}]<\infty .$$
For any $n\in\mathbb N$ and $\pi_n\in\mathcal A_n(x)$, $\{V^{\pi_n}_t\}_{t=0}^T$ is a $\mathbb Q_n$ super--martingale. Hence,
from the Holder inequality (observe that $\frac{1}{p}+\frac{1}{q}=1$) we get
\begin{equation*}
\begin{split}
\sup_{n\in\mathbb N}&\sup_{\pi_n\in\mathcal A_n(x)}\mathbb E_{\mathbb P_n}[(V^{\pi_n}_T)^{1/p}]\\
&=\sup_{n\in\mathbb N}\sup_{\pi_n\in\mathcal A_n(x)}\mathbb E_{{\mathbb Q}_n}\left [(V^{\pi_n}_T)^{1/p}\frac{d\mathbb P_n}{d{\mathbb Q}_n}\right]\\
&\leq \sup_{n\in\mathbb N}\sup_{\pi_n\in\mathcal A_n(x)}
(\mathbb E_{{\mathbb Q}_n}[V^{\pi_n}_T])^{1/p}\sup_{n\in\mathbb N}\left(\mathbb E_{{\mathbb Q}_n}\left[\left(\frac{d\mathbb P_n}
{d{\mathbb Q}_n}\right)^q\right]\right)^{1/q}\\
&\leq x^{1/p}\sup_{n\in\mathbb N}\left(\mathbb E_{{\mathbb Q}_n}\left[\left(\frac{d\mathbb P_n}
{d{\mathbb Q}_n}\right)^q\right]\right)^{1/q}<\infty,
\end{split}
\end{equation*}
and the result follows.
\end{proof}

\section{The necessity of Assumptions \ref{asm5.1},\ref{asm4.1},\ref{asm4.2}}\label{sec:nec}\setcounter{equation}{0}

\subsection{On the necessity of Assumption \ref{asm5.1}}\label{sec:necas3}
Let us explain by example why Assumption \ref{asm5.1} is essential for the lower semi--continuity to hold.
\begin{exm}

\textbf{Naive discretization does not work}.\\
Let $d=1$.
Consider a random utility which corresponds to shortfall risk minimization for a call option with strike price
$K>0$. Namely, we set
\begin{equation}\label{5.utility}
U(v,s):=-((s_T-K)^{+}-v)^{+}.
\end{equation}
We have,
$$u(x)=-\inf_{\pi\in\mathcal A(x)} \mathbb E_{\mathbb P}\left[\left((S_T-K)^{+}-V^\pi_T\right)^{+}\right].$$

Consider the Black--Scholes model
\begin{equation*}
S_t=S_0 e^{\sigma W_t-\sigma^2 t/2}, \ \ t\in [0,T]
\end{equation*}
where $\sigma>0$ is a constant volatility and $W=\{W_t\}_{t=0}^T$ is a Brownian motion (under $\mathbb P$).

We take the naive discretization and define the processes $S^{(n)}$, $n\in\mathbb N$, by
\begin{equation*}
S^{(n)}_t:=S_{\frac{kT}{n}}, \ \ k T/n\leq t<(k+1)T/n.
\end{equation*}
Let $\mathcal F^{(n)}$ the usual filtration which is generated by $S^{(n)}$. Namely,
\begin{equation*}
\mathcal F^{(n)}_t:=\sigma\left\{S_{\frac{T}{n}},...,S_{\frac{kT}{n}},\mathcal N\right\}, \ \ k T/n\leq t<(k+1)T/n
\end{equation*}
where $\mathcal N$ is the collection of all null sets. We also set
$\mathbb P_n:=\mathbb P$.

It is easy to see that $S^{(n)}\Rightarrow S$ and Assumptions \ref{asm2.1}--\ref{asm2.2} hold true (for Assumption \ref{asm2.1+} see Remark \ref{AE:new}).

Next, we check Assumption \ref{asm5.1}. Fix $n$. Recall the processes $D^{(n)},J^{(n)}$ from Assumption \ref{asm5.1}.
First, observe that if $J^{(n)}$ is an adapted left continuous process, then for all $k<n$
$J^{(n)}_{\frac{(k+1)T}{n}}$ is $\mathcal F^{(n)}_{\frac{kT}{n}}$ measurable.
Notice that for or all $k<n$,
$$ess \sup\left(S^{(n)}_{\frac{(k+1)T}{n}}-S^{(n)}_{\frac{kT}{n}}|\mathcal F^{(n)}_{\frac{kT}{n}}\right)=\infty \ \mbox{a.s.}$$
As usual $ess \sup(Y|\mathcal G)$ is the minimal random variable (which may take the value $\infty$)
that is $\mathcal G$ measurable and $\geq Y$ a.s.
These two simple observations yield that
there is no (finite) adapted left continuous process $\{J^{(n)}_t\}_{t=0}^T$
which satisfy
$J^{(n)}_{\frac{(k+1)T}{n}}\geq D^{(n)}_{\frac{(k+1)T}{n}}$. Thus,
Assumption \ref{asm5.1} is not satisfied.

In \cite{M:11} (see Section 6.1.2) it was proved that for the processes $S^{(n)}$, $n\in\mathbb N$ defined above
and the initial capital $x:=\mathbb E_{\mathbb P}[(S_T-K)^{+}]$ (i.e. the Black--Scholes price)
we have
$$\lim\inf_{n\rightarrow\infty} \inf_{\pi_n\in\mathcal A_n(x)} \mathbb E_{\mathbb P}\left[\left((S_T-K)^{+}-V^{\pi_n}_T\right)^{+}\right]>0.
$$
Clearly, the fact that $x$ is the Black--Scholes price implies that
$$\inf_{\pi\in\mathcal A(x)} \mathbb E_{\mathbb P}\left[\left((S_T-K)^{+}-V^\pi_T\right)^{+}\right]=0.$$
We get $$u(x)=0> \lim\sup_{n\rightarrow\infty} u_n(x),$$
and as a result
Proposition \ref{prop5.1} does not hold true.
\end{exm}
\begin{exm}\label{exm.bin}
\textbf{Discrete approximations with vanishing growth rates do work}.\\
Consider a setup where for any $n$, $S^{(n)}$ is a pure jump process of the form
$$S^{(n)}_t=\sum_{i=1}^{m_n} S^{(n)}_{\tau^{(n)}_i}\mathbb I_{\tau^{(n)}_i\leq t<\tau^{(n)}_{i+1}}+S^{(n)}_T\mathbb{I}_{t=T}$$
where $m_n\in\mathbb N$ and
$0=\tau^{(n)}_1<\tau^{(n)}_2<...<\tau^{(n)}_{m_n+1}=T$ are stopping times with respect to $\{\mathcal F^{(n)}_t\}_{t=0}^T$.

Assume that there exists a deterministic sequence
$a_n>0$, $n\in\mathbb N$ such that $\lim_{n\rightarrow\infty} a_n=0$ and
$$|S^{(n)}_{\tau^{(n)}_{i+1}}-S^{(n)}_{\tau^{(n)}_{i}}|\leq a_n | S^{(n)}_{\tau^{(n)}_{i}}| \ \ \mbox{a.s}, \  \ \forall i,n. $$
Then Assumption \ref{asm5.1} holds true with the processes
$$J^{(n)}_t:=a_n \left(\sum_{i=1}^{m_n} \max_{1\leq j\leq i} |S^{(n)}_{\tau^{(n)}_j}| \ \mathbb I_{\tau^{(n)}_i<t\leq\tau^{(n)}_{i+1}}\right), \ \ n\in\mathbb N.$$

In other words, if the growth rates go to zero uniformly then Assumption \ref{asm5.1} holds true.
This is exactly the
case for binomial approximations
of diffusion models with bounded volatility.
\end{exm}

\subsection{On the necessity of Assumption~\ref{asm4.1}}\label{sec:necofas4f}

A natural question to ask is whether Assumption \ref{asm4.1} can be replaced by a simpler one.

In
\cite{HS:98} the authors
analyzed when weak convergence implies the convergence of option prices.
Roughly speaking, the main result was
that under contiguity
properties of the sequences of physical measures
with respect to the martingale measures there
is a convergence of
prices of derivative securities. The contiguity assumption (for the exact definition see \cite{HS:98})
is simpler than Assumption \ref{asm4.1} and deals only with the approximating sequence. The main advantage of such assumption that
it does not require establishing weak convergence (unlike Assumption \ref{asm4.1}).
 However, this classical result assumes that the limit model is complete.
In general, in incomplete markets ``strange phenomena" can happen as we will demonstrate in Example \ref{exm4.2}.

In Example \ref{exm4.2} we construct a sequence of binomial (discrete) martingales $S^{(n)}$ considered with their natural filtrations
that converge weakly to a continuous martingale $S$ (the contiguity assumption trivially holds true).
Surprisingly, the limiting model, which is given by the martingale $S$,
is a fully incomplete market (see Definition 2.1 in \cite{DN:18})
and the set of all equivalent martingale measures is dense in the set of all martingale measures
(for a precise formulation see Lemma 8.1 in \cite{DN:18}).
We use this construction to illustrate that Assumption \ref{asm4.1} is the ``right" assumption to make.

The cornerstone of our construction is the following result which was established in \cite{CHS:16} (see Theorem 8 there). For the reader's
convenience we provide a short self-contained proof.
\begin{lem}\label{lem.kais}
Let $\xi_i=\pm 1$, $i\in\mathbb N$ be i.i.d. and symmetric. Define the processes
$W^{(n)}_t,\hat W^{(n)}_t$, $t\in [0,T]$ by
\begin{eqnarray*}
&W^{(n)}_t:=\sqrt\frac{T}{n}\sum_{i=1}^k \xi_i, \ \ \frac{kT}{n}\leq t<\frac{(k+1) T}{n},\\
&\hat W^{(n)}_t:=\sqrt\frac{T}{n}\sum_{i=1}^k \prod_{j=1}^i \xi_j, \ \ \frac{kT}{n}\leq t<\frac{(k+1) T}{n}
\end{eqnarray*}
where $\sum_{i=1}^0\equiv 0$.
Then, we
have the weak convergence $$(W^{(n)},\hat W^{(n)})\Rightarrow (W,\hat W),$$
where $W$ and $\hat W$ are independent Brownian motions.
\end{lem}
\begin{proof}
We apply the martingale invariance principle given by Theorem 2.1 in \cite{W:07}.
For any $n$ define the filtration $\{\mathcal G^{(n)}_t\}_{t=0}^T$ by
$\mathcal G^{(n)}_t=\sigma\{\xi_1,...,\xi_k\}$ for $k T/n\leq t<(k+1)T/n$. Observe that
$W^{(n)},\hat W^{(n)}$ are martingales with respect to the filtration $\mathcal G^{(n)}$. Thus it remains to establish (2)--(3) in \cite{W:07}.
Clearly,
$$\sup_{0\leq t\leq T} |W^{(n)}_{t}-W^{(n)}_{t-}|=\sup_{0\leq t\leq T} |\hat W^{(n)}_{t}-\hat W^{(n)}_{t-}|=\sqrt\frac{T}{n},$$
and so the maximal jump size goes to zero as $n\rightarrow\infty$. Moreover,
$[W^{(n)}]_t= [\hat {W}^{(n)}]_t=k T/n$ for $k T/n\leq t<(k+1)T/n$. Thus,
$[W^{(n)}]_t\rightarrow t$
and $[\hat W^{(n)}]_t\rightarrow t$ as $n\rightarrow\infty$.

It remains to show that for all $t\in [0,T]$
\begin{equation}\label{large}
[W^{(n)},\hat W^{(n)}]_t\rightarrow 0 \ \mbox{in} \ \mbox{probability}.
\end{equation}
Indeed, let $n\in\mathbb N$ and $k T/n\leq t<(k+1)T/n$. Clearly,
\begin{eqnarray*}
[W^{(n)},\hat W^{(n)}]_t=\frac{T}{n}\sum_{i=1}^k \prod_{j=1}^{i-1} \xi_j, \ \ \frac{kT}{n}\leq t<\frac{(k+1) T}{n},
\end{eqnarray*}
where $\prod_{i=1}^0\equiv 1$.
Observe that the random variables $\prod_{j=1}^{m} \xi_j=\pm 1$, $m\in\mathbb N$ are i.i.d. and symmetric. Thus,
$$\mathbb E\left(\left([W^{(n)},\hat W^{(n)}]_t\right)^2\right)=\frac{T^2}{n^2} k\leq \frac{T t}{n}$$
and (\ref{large}) follows.
This completes the proof.
\end{proof}
\begin{exm}\label{exm4.2}
\textbf{Binomial models can converge weakly to fully incomplete markets.}\\
Let $d=1$.
For any $n\in\mathbb N$ define the stochastic processes
$\{\nu^{(n)}_t\}_{t=0}^T$ and $\{S^{(n)}_t\}_{t=0}^T$ by
\begin{eqnarray*}
&\nu^{(n)}_t:=\prod_{i=1}^k \left(1+\sqrt\frac{T}{n}\xi_i\right), \ \ \frac{kT}{n}\leq t<\frac{(k+1)T}{n},\\
&S^{(n)}_t:=\prod_{i=1}^k \left(1+\min(\nu^{(n)}_{\frac{(i-1)T}{n}},\ln n)\sqrt\frac{T}{n}\prod_{j=1}^i\xi_j\right), \ \ \frac{kT}{n}\leq t<\frac{(k+1)T}{n},
\end{eqnarray*}
where $\xi_i=\pm 1$, $i\in\mathbb N$ are i.i.d. and symmetric. Let $\mathbb P_n$ be the corresponding probability measure.

We assume that $n$ is sufficiently large so that $S^{(n)}$ and $\min(\nu^{(n)},\ln n)$ are strictly positive.
Let
$\mathcal F^{(n)}$ be the filtration which is generated by $S^{(n)}$,
\begin{equation*}
\mathcal F^{(n)}_t:=\sigma\left\{S_{\frac{T}{n}},...,S_{\frac{kT}{n}}\right\}, \ \ k T/n\leq t<(k+1)T/n.
\end{equation*}
Observe that $\mathcal F^{(n)}_t=\sigma\{\xi_1,...,\xi_k\}$ for $k T/n\leq t<(k+1)T/n$. Moreover,
the conditional support of
$supp \left(S^{(n)}_{\frac{(k+1)T}{n}}|S^{(n)}_{\frac{T}{n}},...,S^{(n)}_{\frac{kT}{n}}\right)$ consists of exactly two points, and so
the physical measure $\mathbb P_n$ is the unique martingale measure for $S^{(n)}$.

From
Theorems 4.3--4.4 in \cite{DP:92} and Lemma \ref{lem.kais} we obtain the weak convergence
$(S^{(n)},\nu^{(n)})\Rightarrow (S,\nu)$
where $(S,\nu)$ is the (unique strong) solution of the SDE
\begin{eqnarray}\label{4.700}
&dS_t=\nu_t S_t d\hat W_t, \ \ S_0=1\nonumber\\
&\\
&d\nu_t=\nu_t dW_t, \ \ \nu_0=1\nonumber
\end{eqnarray}
where $W$ and $\hat W$ are independent Brownian motions (under $\mathbb P$).

Namely, for the complete binomial models $S^{(n)}$, $n\in\mathbb N$ we have the weak convergence
$S^{(n)}\Rightarrow S$ where $S$ is the distribution of the stochastic volatility model given by (\ref{4.700}). This is a specific case of the
Hull--White model which was introduced in \cite{HW:87}.
From Theorem 3.3 in \cite{S:98} it follows that $\{S_t\}_{t=0}^T$ is a true martingale.
Hence,
$$\mathbb E_{\mathbb P}[S_T]=S_0=1=S^{(n)}_0=\mathbb E_{\mathbb P_n}[S^{(n)}_T].$$
This together with Theorem 3.6 in \cite{B:68} gives that the random variables $\{S^{(n)}_T\}_{n\in\mathbb N}$
are uniformly integrable.

Let us observe that Assumption \ref{asm4.1} does not hold true.
Indeed, for any $n$ we have the equality
$\mathcal M(S^{(n)})=\left\{\mathbb P_n\right\}$.
Hence, $\left(\left(S,1\right);\mathbb P\right)$
is the only
cluster point for the distributions
$\left(\left(S^{(n)},\frac{d\mathbb Q_n}{d\mathbb P_n}\right);\mathbb P_n\right)$, $\mathbb Q_n\in\mathcal M(S^{(n)})$.
Since the set $\mathcal M(S)$ is not a singleton then clearly Assumption \ref{asm4.1} is not satisfied.

Next, let $K>0$. Consider a call option with strike price $K$ and the utility function given by (\ref{5.utility}).
Obviously, Assumption \ref{asm2.1}(i) and Assumption \ref{asm2.2}(ii) ($U^{+}\equiv 0$) are satisfied.
We want to demonstrate that Proposition \ref{prop4.2} does not hold true.

For any $n\in\mathbb N$ let
$\mathbb V_n$ be the unique arbitrage free price of the above call option in the (complete) model given by $S^{(n)}$.
From the weak convergence $S^{(n)}\Rightarrow S$ and the uniform integrability of $\{S^{(n)}_T\}_{n\in\mathbb N}$ we get
$$
\lim_{n\rightarrow\infty} \mathbb V_n=\mathbb E_{\mathbb P}\left[\left(S_T-K\right)^{+}\right]<S_0=1.
$$
In particular there exists $\epsilon>0$ such that for sufficiently large $n$ we have $\mathbb  V_n<1-\epsilon$. Thus,
$$
\lim_{n\rightarrow\infty} u_n(1-\epsilon)=0.
$$
On the other hand, the model given by $S$ is a fully incomplete market
(see Definition 2.1 and Example 2.5 in \cite{DN:18}).
In \cite{DN:18,N:18} it was proved that in fully incomplete markets the super--replication price
is prohibitively high and lead to buy--and--hold strategies.
Namely, the super--hedging price of a call option is equal to the initial stock price
$S_0=1$. Thus $u(1-\epsilon)<0$ and so
Proposition
\ref{prop4.2} does not hold.
\end{exm}

\subsection{On the necessity of Assumption~\ref{asm4.2}}\label{sec:necofas4s}

\begin{exm}\label{exm4.1}
\textbf{Non-concave utility.}\\
Let $d=1$. Assume that the investor utility function is given by
$$U(v,s):=\min(2,\max(v,1)),$$ and depends only on the wealth.
We notice that the function $U$ does not satisfy Assumption \ref{asm4.2}.

For any $n\in\mathbb N$ consider the binomial model given by
$$S^{(n)}_t:=\prod_{i=1}^k \left(1+\frac{\xi_i}{n^2}\right),  \ \ \frac{kT}{n}\leq t<\frac{(k+1)T}{n},$$
where $\xi_i=\pm 1$, $i\in\mathbb N$ are i.i.d. and symmetric. Namely,
$\mathbb P_n$ is the unique martingale measure for the $n$--th model.
Clearly, for the constant process $S\equiv 1$
we have the weak convergence $S^{(n)}\Rightarrow S$. Thus,
Assumption \ref{asm2.1}(i), Assumption \ref{asm2.2}(ii) and Assumption
\ref{asm4.1} are satisfied.

Next,
consider the initial capital $x:=1$. Observe that for any $n$, there is a set
$A_n\in\sigma\{\xi_1,...,\xi_n\}$ with $\mathbb P_n(A_n)=1/2$. Thus,
from the completeness of the binomial models we get that there exists
$\pi_n\in\mathcal A_n(1)$ such that
$V^{\pi_n}_T=2\mathbb I_{A_n}$.
In particular,
$$u_n(1)\geq \mathbb E_{\mathbb P_n}[\min(2,\max(2\mathbb I_{A_n},1))]=3/2, \ \ n\in\mathbb N.$$
On the other hand, trivially $u(1)=1$, which means that Proposition \ref{prop4.2}
does not hold true.
\end{exm}

The paper \cite{Reichlin:16} studies the continuity
of the value of the utility maximization
problem from terminal wealth (under convergence in distribution)
in a complete market. The author does not assume that the utility function is concave.
The main result says that if the limit probability space is atomless and the atoms in approximating sequence of models
are vanishing (see Assumption 2.1 in \cite{Reichlin:16}) then continuity holds. Clearly, this is not satisfied in the Example
\ref{exm4.1} above where the filtration generated by the limit process is trivial.

An open question is to understand whether the continuity
result from \cite{Reichlin:16} can be extended to the incomplete case.

\section{The Lower Semi--Continuity under Weak Convergence}\label{sec:5}\setcounter{equation}{0}
In this section we prove Proposition \ref{prop5.1}. We start by
establishing a general result.

For any $M>0$ and $n\in\mathbb N$
introduce the set $\Gamma^{(n)}_M$ of all simple predictable integrands of the from
$$
\gamma^{(n)}_t=\sum_{i=1}^k \beta_i\mathbb{I}_{t_i<t\leq t_{i+1}}
$$
where
$k\in\mathbb N$, $0=t_1<t_2<....<t_{k+1}=T$
is a deterministic partition and
\begin{equation*}
\beta_i=\psi_i(S^{(n)}_{a_{i,1}},...,S^{(n)}_{a_{i,m_i}}), \ \ i=1,...,k,
\end{equation*}
for a deterministic partition
$0=a_{i,1}<...<a_{i,m_i}=t_{i}$ and a continuous function
$\psi_i:(\mathbb R^d)^{m_i}\rightarrow\mathbb R^d$ that satisfies $|\psi_i|\leq M$.
\begin{lem}\label{lem5.+}
Let
$\gamma$
be a predictable
process (with respect to $(\mathcal F^S_t)_{0\leq t\leq T}$)
with $|\gamma|\leq M$ for some constant $M$.
 Then there exists
a sequence
$\gamma^{(n)}\in\Gamma^{(n)}_M$, $n\in\mathbb N$ such that
we have the weak convergence
\begin{equation}\label{5.4+}
 \left(\{S^{(n)}_t\}_{t=0}^T,\left\{\int_{0}^t\gamma^{(n)}_{u} dS^{(n)}_{u}\right\}_{t=0}^T\right)\Rightarrow
  \left(\{S_t\}_{t=0}^T,\left\{\int_{0}^t\gamma_{u} dS_{u}\right\}_{t=0}^T\right)
\end{equation}
on the space $\mathbb D([0,T];\mathbb R^d)\times \mathbb D([0,T];\mathbb R)$.
\end{lem}
\begin{proof}
On the space $(\Omega,\mathcal F,(\mathcal F^S_t)_{0\leq t\leq T},\mathbb P)$,
let $\Gamma_M$ be the set of all integrands of the form
\begin{equation}\label{5.0-}
\gamma_t=\sum_{i=1}^k \beta_i\mathbb{I}_{t_i<t\leq t_{i+1}}
\end{equation}
where $k\in\mathbb N$,
$0=t_1<t_2<....<t_{k+1}=T$
is a deterministic partition and
\begin{equation}\label{5.0}
\beta_i=\psi_i(S_{a_{i,1}},...,S_{a_{i,m_i}}), \ \ i=1,...,k
\end{equation}
for a deterministic partition
$0=a_{i,1}<...<a_{i,m_i}=t_{i}$ and a continuous function
$\psi_i:(\mathbb R^d)^{m_i}\rightarrow\mathbb R^d$
which satisfy $|\psi_i|\leq M$. From standard density arguments it follows that
for any $\epsilon>0$ we can find $\gamma'\in\Gamma_M$
which satisfy
$$\mathbb P\left(\sup_{0\leq t\leq T}\left|\int_{0}^t\gamma_{u} dS_{u}-\int_{0}^t
\gamma'_{u} dS_{u}\right|>\epsilon\right)<\epsilon.$$
Hence, without loss of generality we can assume that $\gamma\in \Gamma_{M}$. Thus, let $\gamma$ be given by (\ref{5.0-})--(\ref{5.0}).

For any $n\in\mathbb N$ define $\gamma^{(n)}\in\Gamma^{(n)}_M$
by
\begin{equation}\label{5.1}
\gamma^{(n)}_t:=\sum_{i=1}^k \psi_i\left(S^{(n)}_{a_{i,1}},...,S^{(n)}_{a_{i,m_i}}\right)\mathbb{I}_{t_i<t\leq t_{i+1}}, \ \ t\in [0,T].
\end{equation}
It is well known that there exists a metric $d$ on the Skorokhod space $\mathbb D([0,T];\mathbb R^d)$
that induces the Skorokhod topology and such that $\mathbb D([0,T];\mathbb R^d)$ is separable under $d$
(for details see
Chapter 3 in \cite{B:68}).
From the weak convergence $S^{(n)}\Rightarrow S$ and the Skorokhod representation theorem (see Theorem 3 in \cite{D:68})
 it follows that we can redefine the stochastic processes
$S^{(n)}$, $n\in\mathbb N$ and $S$ on the same probability space such that
$\lim_{n\rightarrow\infty} d(S^{(n)},S)=0$ a.s.
Recall that if $\lim_{n\rightarrow\infty} d(z^{(n)},z)=0$ and $z:[0,T]\rightarrow\mathbb R^d$ is a continuous
function then $\lim_{n\rightarrow\infty} \sup_{0\leq t\leq T}|z^{(n)}_t-z_t|=0$ (see  e.g. Chapter 3 in \cite{B:68}).
We conclude that
\begin{equation}\label{5.3+}
\sup_{0\leq t\leq T}|S^{(n)}_t-S_t|\rightarrow 0 \ \mbox{a.s.}
\end{equation}
Next, recall the partition $0=t_1<t_2<....<t_{k+1}=T$ and redefine (on the common probability space)
 the integrands $\gamma$, $\gamma^{(n)}$ by the relations (\ref{5.0-})--(\ref{5.1}).
 Since these integrands are simple then the corresponding stochastic integrals
 $\int_{0}^t \gamma_{u} dS_{u}, \int_{0}^t \gamma^{(n)}_{u} dS^{(n)}_{u}$, $t\in [0,T]$
 can be redefined as finite sums.

From (\ref{5.3+}) and the continuity of $\psi_i$, $i=1,...,k$ we get
that
$$ \sup_{0\leq t\leq T}|\gamma^{(n)}_{t}-\gamma_{t}|\rightarrow 0 \ \ \mbox{a.s.}$$
Thus,
\[
\begin{split}
\sup_{0\leq t\leq T}&\left|\int_{0}^t \gamma^{(n)}_{u} dS^{(n)}_{u}-\int_{0}^t \gamma_{u} dS_{u} \right| \\
&=\sup_{0\leq t\leq T}\left|\sum_{i=1}^k \left(\gamma^{(n)}_{t_{i+1}}(S^{(n)}_{t_{i+1}\wedge t}-S^{(n)}_{t_i\wedge t})-
\gamma_{t_{i+1}}(S_{t_{i+1}\wedge t}-S_{t_i\wedge t})\right)\right|\\
&\leq\sup_{0\leq t\leq T}\left|\sum_{i=1}^k \gamma^{(n)}_{t_{i+1}}\left((S^{(n)}_{t_{i+1}\wedge t}-S^{(n)}_{t_i\wedge t})
-(S_{t_{i+1}\wedge t}-S_{t_i\wedge t})\right)\right|\\
&+\sup_{0\leq t\leq T}\left|\sum_{i=1}^k (\gamma^{(n)}_{t_{i+1}}-\gamma_{t_{i+1}})(S_{t_{i+1}\wedge t}-S_{t_i\wedge t})\right|\\
&\leq 2 M k d
\sup_{0\leq t\leq T}|S^{(n)}_t-S_t|\\
&+2 k  d  \sup_{0\leq t\leq T}|\gamma^{(n)}_{t}-\gamma_{t}|\sup_{0\leq t\leq T} |S_t| \rightarrow 0 \ \ \mbox{a.s.}
\end{split}
\]
and the proof is completed.
\end{proof}
Now, we are ready to prove Proposition \ref{prop5.1}. \\
\textbf{Proof of Proposition \ref{prop5.1}}.\\
The proof will be done in two steps.\\
\textbf{Step I:}
For any $x>0$ let $\mathcal A_0(x)\subset\mathcal A(x)$ be the set of all
admissible portfolios $\pi=(x,\gamma)$ such that $\gamma$
is predictable, uniformly bounded and of bounded variation.
In this step we show that for any $x_1>x_2>0$
\begin{equation}\label{5.5}
u(x_2)\leq\sup_{\pi\in\mathcal A_0(x_1)}
\mathbb E_{\mathbb P}[U(V^\pi_T,S)].
\end{equation}
A priori the left hand side and the right hand side of (\ref{5.5}) can be both equal to $\infty$.

Let $\bar{\pi}=(x_2,\bar{\gamma})\in\mathcal A(x_2)$ be an arbitrary portfolio.
By applying the density argument given by Theorem 3.4 in \cite{BankBaum:04} we obtain that there exists an adapted continuous process of bounded
variation
$\tilde\gamma=\{\tilde\gamma_t\}_{t=0}^T$ such that
$$\sup_{0\leq t\leq T}\left|\int_{0}^t \tilde\gamma_{u} dS_{u}-\int_{0}^t \bar\gamma_{u} dS_{u}\right|\leq \frac{x_1-x_2}{2} \ \ \mbox{a.s.}$$
We conclude that the portfolio which is given by $\tilde\pi:=(x_1,\tilde\gamma)$
satisfies
\begin{equation}\label{5.new}
V^{\tilde\pi}_t\geq V^{\bar\pi}_t+\frac{x_1-x_2}{2}\geq\frac{x_1-x_2}{2}, \ \ t\in [0,T].
\end{equation}

Next, for the continuous process $\tilde\gamma$
define the stopping times
$$\theta_n:=T\wedge\inf\{t:|\tilde\gamma_t|= n\}, \ \ n\in\mathbb N$$
and the trading strategies $${\tilde\gamma}^{(n)}_t:=\mathbb{I}_{t\leq\theta_n}\tilde\gamma_t, \ \ t\in [0,T].$$
Set $\tilde\pi_n=(x_1,\tilde\gamma^{(n)})$. Clearly, $|\tilde\gamma^{(n)}|\leq n$ and from (\ref{5.new}) we have
$$V^{{\tilde\pi}_n}_t=V^{\tilde\pi}_{t\wedge\theta_n}\geq\frac{x_1-x_2}{2}, \ \ t\in [0,T].$$
 Hence, ${\tilde\pi}_n\in\mathcal A_0(x_1)$. Observe that
$\theta_n\uparrow T$ a.s., and so
$$\lim_{n\rightarrow\infty} V^{{\tilde\pi}_n}_T=\lim_{n\rightarrow\infty} V^{\tilde\pi}_{\theta_n}=V^{\tilde\pi}_T.$$
This together with Fatou's Lemma, Assumption \ref{asm2.1} (notice that $V^{{\tilde\pi}_n}_T\geq \frac{x_1-x_2}{2}>0$),
the fact that $U$ is continuous and (\ref{5.new}) gives
 $$\mathbb E_{\mathbb P}[U(V^{\bar\pi}_T,S)]\leq\mathbb E_{\mathbb P}[U(V^{{\tilde\pi}}_T,S)]\leq\lim\inf_{n\rightarrow\infty} \mathbb E_{\mathbb
 P}[U(V^{\tilde{\pi}_n}_T,S)].$$
Since $\bar\pi\in\mathcal A(x_2)$ was arbitrary we complete the proof of (\ref{5.5}).\\
\textbf{Step II:}
In view of (\ref{5.5})
and Assumption \ref{asm2.1+}, in order to prove Proposition \ref{prop5.1} it sufficient to show that for
any initial capital $x>0$, $0<\epsilon<\frac{x}{2}$ and $\pi\in\mathcal A_0(x-2\epsilon)$ there exists a sequence
$\pi_n\in\mathcal A_n(x)$, $n\in\mathbb N$ which satisfies
\begin{equation}\label{5.6}
\lim\inf_{n\rightarrow\infty} \mathbb E_{\mathbb P_n}[U(V^{\pi_n}_T,S^{(n)})]\geq \mathbb E_{\mathbb P}[U(V^\pi_T,S)].
\end{equation}
Let $0<\epsilon<\frac{x}{2}$ and
$\pi=(x-2\epsilon,\gamma)\in\mathcal A_0(x-2\epsilon)$. Let $M>0$ such that $|\gamma|\leq M$.
Lemma \ref{lem5.+} provides the existence of simple integrands
$\gamma^{(n)}\in\Gamma^{(n)}_M$, $n\in\mathbb N$ which satisfy (\ref{5.4+}).

For a given $n$, the portfolio $(x,\gamma^{(n)})$
might fail to be admissible and so
a modification is needed.
Recall Assumption \ref{asm5.1} and the stochastic processes $J^{(n)}$, $n\in\mathbb N$.
For any $n\in\mathbb N$ introduce the stopping time
\begin{equation}\label{5.1000}
\Theta_n:=T\wedge\inf\left\{t:x+\int_{0}^{t-}\gamma^{(n)}_{u}dS^{(n)}_{u}<\epsilon+ M d J^{(n)}_{t}\right\},
\end{equation}
and define the portfolio $\pi_n=(x,\bar\gamma^{(n)})$ by
$\bar\gamma^{(n)}_t:=\mathbb{I}_{t\leq\Theta_n}\gamma^{(n)}_t.$
Let us show that $V^{\pi_n}_t\geq\epsilon$ for all $t\in [0,T]$. Indeed,
\[
\begin{split}
V^{\pi_n}_t&=x+\int_{0}^{t\wedge\Theta_n}\gamma^{(n)}_{u}dS^{(n)}_{u}\\
& \geq x+\int_{0}^{{t\wedge\Theta_n}-}\gamma^{(n)}_{u}dS^{(n)}_{u}- M d
|S^{(n)}_{{t\wedge\Theta_n}-}-S^{(n)}_{t\wedge\Theta_n}|\\
& \geq \epsilon+M d J^{(n)}_{t\wedge\Theta_n}-M d |S^{(n)}_{\Theta}-S^{(n)}_{\Theta-}|\geq\epsilon
\end{split}
\]
as required.
The first inequality follows from $|\gamma^{(n)}|\leq M$.
The second inequality follows from the fact that on the time interval $[0,\Theta_n)$ se have
$x+\int_{0}^{\cdot-}\gamma^{(n)}_{u}dS^{(n)}_{u}\geq \epsilon+ M d J^{(n)}_{\cdot}.$
The last inequality is due to
$J^{(n)}_{\Theta}\geq |S^{(n)}_{\Theta}-S^{(n)}_{\Theta-}|$.
We conclude that $\pi_n\in\mathcal A_n(x)$ and
\begin{equation}\label{5.1002}
V^{\pi_n}_T=x+\int_{0}^{\Theta_n}\gamma^{(n)}_{u}dS^{(n)}_{u}\geq\epsilon.
\end{equation}
Next, we apply the Skorokhod representation theorem.
Recall that the processes $\{J^{(n)}_t\}_{t=0}^T$, $n\in\mathbb N$ are non--negative, non decreasing and
$J^{(n)}_T\rightarrow 0$ in probability. This together with (\ref{5.4+}) implies that we have the weak convergence
\begin{equation}\label{5.40+}
 \left(\{J^{(n)}_t\}_{t=0}^T,\{S^{(n)}_t\}_{t=0}^T,\left\{\int_{0}^t\gamma^{(n)}_{u} dS^{(n)}_{u}\right\}_{t=0}^T\right)\Rightarrow
  \left(0,\{S_t\}_{t=0}^T,\left\{\int_{0}^t\gamma_{u} dS_{u}\right\}_{t=0}^T\right)
\end{equation}
on the space $\mathbb D([0,T];\mathbb R)\times\mathbb D([0,T];\mathbb R^d)\times \mathbb D([0,T];\mathbb R)$.

For any $n\in\mathbb N$ the integrand $\gamma^{(n)}$ is of the form (\ref{5.1}). Hence the integrand
$\gamma^{(n)}$ and the corresponding stochastic integral $\int_{0}^{\cdot}\gamma^{(n)}_{u} dS^{(n)}_{u}$
are determined pathwise by $S^{(n)}$. Since $\gamma$ is of bounded variation then we have
$$\int_{0}^t \gamma_{u} dS_{u}=\gamma_t S_t-\gamma_0 S_0-\int_{0}^t S_{u} d\gamma_{u},$$
where the last term is the pathwise Riemann--Stieltjes integral. We conclude that
$\gamma$ and the corresponding stochastic integral $\int_{0}^{\cdot}\gamma_{u} dS_{u}$
are determined pathwise by $S$.

Thus, from the Skorokhod representation theorem and (\ref{5.40+}) it follows that we can redefine the stochastic processes
$\gamma^{(n)},S^{(n)},J^{(n)}$, $n\in\mathbb N$ and $\gamma,S$ on the same probability space such that (\ref{5.3+}) holds true,
\begin{equation}\label{5.nnn}
\sup_{0\leq t\leq T} J^{(n)}_t\rightarrow 0 \ \ \mbox{a.s.}
\end{equation}
and
\begin{equation}\label{5.8}
\sup_{0\leq t\leq T}\left|\int_{0}^t\gamma^{(n)}_u dS^{(n)}_u-\int_{0}^t\gamma_u dS_u\right|\rightarrow 0 \ \mbox{a.s.}
\end{equation}
As in (\ref{5.3+}) the uniform convergence is due to the fact that the limit processes are continuous.
By applying (\ref{5.1000}) we redefine
$\Theta_n$, $n\in\mathbb N$ on the common probability space.
With abuse of notations we denote by $\mathbb P$ and $\mathbb E$ the probability and the expectation on the common probability space, respectively.

First, we argue that
\begin{equation}\label{5.nn}
\lim_{n\rightarrow\infty}\mathbb P(\Theta_n=T)=1.
\end{equation}
Recall, the admissible portfolio $\pi=(x-2\epsilon,\gamma)$.
From (\ref{5.8}) it follows that
$$\lim\inf_{n\rightarrow\infty}\inf_{0\leq t\leq T}\left(x+\int_{0}^ t\gamma^{(n)}_u dS^{(n)}_u\right)=x+\inf_{0\leq t\leq T}\int_{0}^t \gamma_u dS_u\geq 2 \epsilon.$$
In particular
$$
\lim_{n\rightarrow\infty} \mathbb P\left(\inf_{0\leq t\leq T}\left(x+\int_{0}^ t\gamma^{(n)}_u dS^{(n)}_u\right)>\frac{3\epsilon}{2}\right)=1.
$$
This together with (\ref{5.nnn}) gives (\ref{5.nn}).

Finally, from Fatou's Lemma, the continuity of $U$, Assumption \ref{asm2.1}(i), Assumption \ref{asm2.2}(i) (recall that $V^{\pi_n}_T\geq\epsilon$),
(\ref{5.3+}), (\ref{5.1002}) and (\ref{5.8})--(\ref{5.nn}) we obtain
\[
\begin{split}
\lim\inf_{n\rightarrow\infty} \mathbb E_{\mathbb P_n}[U(V^{\pi_n}_T,S^{(n)})]
&=\lim\inf_{n\rightarrow\infty} \mathbb E\left[U\left(x+\int_{0}^{\Theta_n} \gamma^{(n)}_u dS^{(n)}_u,S^{(n)}\right)\right]\\
&\geq \mathbb E\left[U\left(x+\int_{0}^{T} \gamma_u dS_u,S\right)\right]\geq\mathbb E_{\mathbb P}[U(V^\pi_T,S)],
\end{split}
\]
and (\ref{5.6}) follows.
\hfill $\square$

\section{The Upper Semi--Continuity under Weak Convergence}\label{sec:4}\setcounter{equation}{0}
In this section we prove Proposition \ref{prop4.2}.
\begin{proof}
Let $x>0$. From Assumption \ref{asm2.2}(ii) it follows that for any $n\in\mathbb N$ $u_n(x)<\infty$.
Hence, we can choose a sequence $\hat\pi_n\in\mathcal{A}_n(x)$, $n\in\mathbb N$ which satisfy
(\ref{2.asp}).
Without loss of generality (by passing to a subsequence) we assume that the limit
$\lim_{n\rightarrow\infty} \mathbb E_{\mathbb P_n}[U(V^{\hat\pi_n}_T,S^{(n)})]$ exists.
We will prove that there exists $\hat\pi\in\mathcal A(x)$
such that
\begin{equation}\label{4.1000}
\mathbb E_{\mathbb P}[U(V^{\hat\pi}_T,S)]\geq\lim_{n\rightarrow\infty} \mathbb E_{\mathbb P_n}[U(V^{\hat\pi_n}_T,S^{(n)})]
\end{equation}
and this will give Proposition \ref{prop4.2}.
The proof will be done in two steps. \\
\textbf{Step I:}
Choose $\mathbb Q\in\mathcal M(S)$ (recall that we assume $\mathcal M(S)\neq\emptyset$) and
set $Z:=\frac{d\mathbb Q}{d\mathbb P}$.
From Assumption \ref{asm4.1} it follows that there exists a sequence $\mathbb Q_n\in\mathcal M(S^{(n)})$, $n\in\mathbb N$ for which (\ref{4.1}) holds true. For any $n$,
$\{V^{\hat\pi_n}_t\}_{t=0}^T$ is a $\mathbb Q_n$ super--martingale.
Hence,
$$\mathbb E_{\mathbb P_n}\left(V^{\hat\pi_n}_T\frac{d\mathbb Q_n}{d\mathbb P_n}\right)=
\mathbb E_{\mathbb Q_n}[V^{\hat\pi_n}_T]
\leq V^{\hat\pi_n}_0=x.$$
We conclude that the sequence
$\left(V^{\hat\pi_n}_T\frac{d\mathbb Q_n}{d\mathbb P_n};\mathbb P_n\right)$, $n\in\mathbb N$ is tight. This together with
(\ref{4.1}) yields that
the sequence
$\left(\left(S^{(n)},\frac{d\mathbb Q_n}{d\mathbb P_n},V^{\hat\pi_n}_T\frac{d\mathbb Q_n}{d\mathbb P_n}\right);\mathbb P_n\right)$,
$n\in\mathbb N$
is tight on the space $\mathbb D([0,T];\mathbb R^d)\times \mathbb R^2$. From Prohorov's theorem
it follows that there exists
a subsequence
$\left(\left(S^{(n)},\frac{d\mathbb Q_n}{d\mathbb P_n},V^{\hat\pi_n}_T
\frac{d\mathbb Q_n}{d\mathbb P_n}\right);\mathbb P_n\right)$
(for simplicity
 the subsequence is still denoted by $n$) which converge weakly.
From (\ref{4.1}) we obtain that
\begin{equation}\label{4.3--}
\left(\left(S^{(n)},\frac{d\mathbb Q_n}{d\mathbb P_n},V^{\hat\pi_n}_T
\frac{d\mathbb Q_n}{d\mathbb P_n}\right);\mathbb P_n\right)\Rightarrow (S,Z,Y),
\end{equation}
where $Y$ is some random variable.
In particular we have the weak convergence
\begin{equation}\label{4.3-}
\left(\left(S^{(n)},V^{\hat\pi_n}_T
\right);\mathbb P_n\right)\Rightarrow \left(S,\frac{Y}{Z}\right).
\end{equation}
The random vector $(S,Z,Y)$ is defined on a new probability space
$(\tilde\Omega,\tilde{\mathcal F},\tilde{\mathbb P})$,
which might be different from the original probability space
$(\Omega,{\mathcal F},{\mathbb P})$.
We redefine the filtration $\mathcal F^S$ (the usual filtration
which is generated by $S$) and the sets $\mathcal M(S),\mathcal A(\cdot)$ (as before, these sets defined
with respect to $\mathcal F^S$) on the new probability space
$(\tilde\Omega,\tilde{\mathcal F},\tilde{\mathbb P})$.

Set (notice that $\frac{Y}{Z}\geq 0$)
\begin{equation}\label{4.3}
V:=\mathbb E_{\tilde{\mathbb P}}\left(\frac{Y}{Z}\left|\right.\mathcal F^S_T\right)
\end{equation}
where a priori $V$ can be equal to $\infty$ with finite probability.
In order to prove (\ref{4.1000}) it is sufficient to show that there exists $\hat\pi\in\mathcal A(x)$ such that
\begin{equation}\label{opt}
V^{\hat\pi}_T\geq V \ \mbox{a.s.}
\end{equation}
Indeed, if (\ref{opt}) holds true (in particular $V<\infty$ a.s.), then from
the Jensen inequality, the continuity of $U$, Assumption \ref{asm2.1}(i), Assumption \ref{asm2.2}(ii), Assumption \ref{asm4.2} and (\ref{4.3-}) we obtain
\begin{equation}\label{4.1001}
\mathbb E_{{\mathbb P}}[U(V^{\hat\pi}_T,S)]\geq\mathbb E_{\tilde{\mathbb P}}[U(V,S)]\geq \mathbb E_{\tilde{\mathbb P}}\left[U\left(\frac{Y}{Z},S\right)\right]
\geq \lim_{n\rightarrow\infty} \mathbb E_{\mathbb P_n}[U(V^{\hat\pi_n}_T,S^{(n)})]
\end{equation}
as required.

This brings us to the second step.\\
\textbf{Step II:}
In this step we establish (\ref{opt}).
In view of the optional decomposition theorem
(Theorem 3.2 in \cite{Kram:96}) it is sufficient to show that the super-hedging price which is given by
$\sup_{\hat{\mathbb Q}\in\mathcal M(S)}\mathbb E_{\hat{\mathbb Q}} [V]$ is less or equal than $x$.
From (\ref{4.3}) we obtain
$$\sup_{\hat{\mathbb Q}\in\mathcal M(S)}\mathbb E_{\hat{\mathbb Q}} [V]=
\sup_{\hat{\mathbb Q}\in\mathcal M(S)}\mathbb E_{\tilde{\mathbb P}} \left[\frac{Y}{Z}\frac{d\hat{\mathbb Q}}{d\tilde{\mathbb P}}\right].
$$
Hence, it remains to prove that for any $\hat{\mathbb Q}\in\mathcal M(S)$
\begin{equation}\label{4.5}
x\geq \mathbb E_{\tilde{\mathbb P}} \left[\frac{Y}{Z}\frac{d\hat{\mathbb Q}}{d\tilde{\mathbb P}}\right].
\end{equation}
From Assumption
\ref{asm4.1} we get a sequence $\hat{\mathbb Q}_n\in\mathcal M(S^{(n)})$, $n\in\mathbb N$ for which
\begin{equation}\label{4.5+}
\left(\left(S^{(n)},\frac{d\hat{\mathbb Q}_n}{d\mathbb P_n}\right);\mathbb P_n\right)\Rightarrow
\left(\left(S,\frac{d\hat{\mathbb Q}}{d\mathbb P}\right);\mathbb P\right).
\end{equation}
This together with (\ref{4.3--}) yields that the sequence
$$\left(\left(S^{(n)},\frac{d\mathbb Q_n}{d\mathbb P_n},V^{\pi_n}_T
\frac{d\mathbb Q_n}{d\mathbb P_n},\frac{d\hat{\mathbb Q}_n}{d\mathbb P_n}\right);\mathbb P_n\right), \ \ n\in\mathbb N,$$
is tight on the space $\mathbb D([0,T];\mathbb R^d)\times\mathbb R^3$.
From Prohorov's Theorem and (\ref{4.3--}) there is a subsequence which converge weakly
\begin{equation}\label{4.6}
\left(\left(S^{(n)},\frac{d\mathbb Q_n}{d\mathbb P_n},V^{\hat\pi_n}_T
\frac{d\mathbb Q_n}{d\mathbb P_n},\frac{d\hat{\mathbb Q}_n}{d\mathbb P_n}\right);\mathbb P_n\right)\Rightarrow (S,Z,Y,X)
\end{equation}
for some random variable $X$.

Once again, the random vector $(S,Z,Y,X)$ is defined on a new probability space
$(\tilde{\tilde\Omega},\tilde{\tilde{\mathcal F}},\tilde{\tilde{\mathbb P}})$, on which
we redefine the filtration $\mathcal F^S$ and the sets $\mathcal M(S),\mathcal A(\cdot)$.

Observe that
$\frac{d\hat{\mathbb Q}}{d\mathbb P}$ is determined by $S$. Hence, there exists a measurable function
$g:\mathbb D([0,T];\mathbb R^d)\rightarrow\mathbb R$ such that
$\frac{d\hat{\mathbb Q}}{d\mathbb P}=g(S)$ \ $\mathbb P$ a.s, i.e.
$\mathbb E_{\mathbb P}|\frac{d\hat{\mathbb Q}}{d\mathbb P}-g(S)|=0$.
From (\ref{4.5+})--(\ref{4.6}) we get that the distribution of $(S,X)$ equals to
 $\left(\left(S,\frac{d\hat{\mathbb Q}}{d\mathbb P}\right);\mathbb P\right)$. Thus,
  $\mathbb E_{\tilde{\tilde{\mathbb P}}}|X-g(S)|=0$. We conclude that $X=g(S)$  $\tilde{\tilde{\mathbb P}}$ a.s.

  Finally, from Fatou's Lemma, (\ref{4.6}) and the fact that $\{V^{\hat\pi_n}_t\}_{t=0}^T$ is a $\hat{\mathbb Q}_n$ super--martingale
 it follows that
\[
\begin{split}
 \mathbb E_{\tilde{\mathbb P}} \left(\frac{Y}{Z}\frac{d\hat{\mathbb Q}}{d\tilde{\mathbb P}}\right)&=
 \mathbb E_{\tilde{\mathbb P}} \left(\frac{Y}{Z}g(S)\right)=
 \mathbb E_{\tilde{\tilde{\mathbb P}}} \left(\frac{Y g(S)}{Z}\right)=
 \mathbb E_{\tilde{\tilde{\mathbb P}}} \left(\frac{Y X}{Z}\right)\\
 &\leq \lim\inf_{n\rightarrow\infty}
 \mathbb E_{\mathbb P_n}\left(V^{\hat\pi_n}_T\frac{d{\mathbb Q_n}}{d\mathbb P_n}\frac{\frac{d\hat{\mathbb Q}_n}
 {d\mathbb P_n}}{\frac{d{\mathbb Q_n}}{d\mathbb P_n}}\right)
 =
\lim\inf_{n\rightarrow\infty}
 \mathbb E_{\hat{\mathbb Q}_n}[V^{\hat\pi_n}_T]\leq x,
\end{split}
\]
from which we get (\ref{4.5}).
\end{proof}
Next, we prove
Theorem \ref{thm2.2}.

\subsection{Proof of Theorem \ref{thm2.2}}\label{sec:4WEAK}
In order to prove the statement it is sufficient to
show the for any sub--sequence of laws
$\left(S^{(n)},V^{\hat\pi_n}_T\right)$ there is a further subsequence which converge weakly to
$\left(S,V^{\hat\pi}_T\right)$.

Following the same arguments as in the proof of Proposition \ref{prop4.2} we obtain that for any subsequence of laws
$\left(S^{(n)},V^{\hat\pi_n}_T\right)$
there is a further sequence
which satisfies (\ref{4.3-}).
Moreover, there exists
$\hat\pi\in\mathcal A(x)$ such that (\ref{opt})--(\ref{4.1001}) hold true.

From (\ref{4.1001}), Theorem \ref{main}
and the fact that $\hat\pi_n\in\mathcal A_n(x)$ are asymptotically optimal we get
$$u(x)\geq\mathbb E_{{\mathbb P}}[U(V^{\hat\pi}_T,S)]\geq\mathbb E_{\tilde{\mathbb P}}
\left[U\left(\mathbb E_{\tilde{\mathbb P}}\left(\frac{Y}{Z}\left|\right.\mathcal F^S_T\right),S\right)\right]
\geq \mathbb E_{\tilde{\mathbb P}}\left[U\left(\frac{Y}{Z},S\right)\right]
\geq u(x).$$
We conclude that all the above inequalities are in fact equalities.
This equality together with (\ref{opt}) and the assumption that $U$ is strictly concave and strictly increasing in the first variable (follows from Assumption \ref{asm2.1}(i) and the strict concavity) imply that
$$V^{\hat\pi}_T=\mathbb E_{\tilde{\mathbb P}}\left(\frac{Y}{Z}\left|\right.\mathcal F^S_T\right)=\frac{Y}{Z}$$ and
$\hat\pi\in\mathcal A(x)$ is the unique optimal portfolio. This completes the proof.
\hfill $\square$

We end this section with a remark on how our results can be generalized.
\begin{rem}
Consider the case where the filtration $\mathcal F:=\mathcal F^{S,Y}$ is the usual filtration generated by
$S$ and another RCLL process
$R=(R^{(1)}_t,...,R^{(m)}_t)_{0\leq t\leq T}$. The process $R$ can be viewed as a collection of non traded assets.

For the approximate model we take $(S^{(n)},R^{(n)})$ and a filtration which satisfies the usual assumptions and makes both $S^{(n)}$ and
$R^{(n)}$ adapted. Once again $R^{(n)}=(R^{n,1}_t,...,R^{n,m}_t)_{0\leq t\leq T}$ is the collection of non traded assets.
Consider a continuous utility function
$U:(0,\infty) \times\mathbb D([0,T];\mathbb R^d)\times \mathbb D([0,T];\mathbb R^m)\rightarrow\mathbb R$
and assume the weak convergence $(S^{(n)},R^{(n)})\Rightarrow (S,R)$ and
an analogous assumptions to
those in Section \ref{sec:2}.
Of course, as before the martingale measures are with respect to the traded assets.
 Then, by using similar arguments as in Sections \ref{sec:5}--\ref{sec:4} we can extend the main results Theorems \ref{main}-\ref{thm2.2} to this setup as well.
 \end{rem}

\section{Lattice Based Approximations of the Heston Model}\label{sec:7}
Consider the Heston model \cite{Hest:93} given by
\begin{eqnarray*}
&d\hat S_t=\hat S_t(\mu dt+\sqrt{\hat \nu_t}dW_t),\\
&d\hat\nu_t=\kappa(\theta-{\hat\nu_t})dt+\sigma \sqrt{\hat \nu_t}d\tilde W_t,\nonumber
\end{eqnarray*}
where $\mu\in\mathbb R$, $\kappa,\theta,\sigma>0$ are constants and $W$,
$\tilde W$ are two standard Brownian motions with a constant
correlation $\rho\in (-1,1)$. The initial values $\hat S_0,\hat\nu_0>0$ are given. We assume the condition
$2\kappa\theta>\sigma^2$ which guarantees that $\hat\nu$ does not touch zero (see \cite{CIR:85}).

For technical reasons our approximations require that the volatility will lie in an interval
of the form $[\underline{\sigma},\overline{\sigma}]$ for some $0<\underline{\sigma}<\overline{\sigma}$.
Thus, we modify the Heston model as following.
Fix two barriers $0<\underline{\sigma}<\overline{\sigma}$ and define the function
$$h(z):=\max(\underline{\sigma}^2,\min(z,\overline{\sigma}^2)),\ \ z\in \mathbb R. $$
Consider the SDE
\begin{eqnarray}
&dS_t=S_t(\mu dt+\sqrt{h(\nu_t)} dW_t)\nonumber\\
&{}\label{7.0-}\\
&d\nu_t=\kappa(\theta-h(\nu_t))dt+\sigma \sqrt{h(\nu_t)}d\tilde W_t\nonumber
\end{eqnarray}
where the initial values are $S_0:=\hat S_0$, $\nu_0:=\hat\nu_0$.
Observe that $\sqrt h,h$ are Lipschitz continuous, and so (\ref{7.0-}) has a unique solution.

We expect that if $\underline{\sigma}$ is small and $\overline{\sigma}$ is large then the value of the utility maximization problem in
the Heston model will be close
to the one in the model given by (\ref{7.0-}).
For the shortfall risk measure we provide an error estimate in Lemma \ref{error}.

\subsection{Discretization}\label{sec:dis}
In this section we construct  discrete time lattice based approximations for the model given by (\ref{7.0-}).
The novelty of our constructions is that the approximating sequence satisfies
Assumptions \ref{asm5.1},\ref{asm4.1}.

It is more convenient
to work with a transformed system of equations driven by independent Brownian motions.
 Therefore, we set
$$\Phi_t:=\ln {S_t}, \ \ \Psi_t:=\frac{\nu_t}{\sigma}-\rho \Phi_t.$$
From It\^o's formula we obtain that
\begin{eqnarray*}
&d\Phi_t=\mu_{\Phi}(\Phi_t,\Psi_t) dt+\sigma_{\Phi}(\Phi_t,\Psi_t)dW_t\\
&d\Psi_t=\mu_{\Psi}(\Phi_t,\Psi_t) dt+\sigma_{\Psi}(\Phi_t,\Psi_t) d\hat W_t
\end{eqnarray*}
where
\begin{eqnarray*}
&\mu_{\Phi}(y,z):=\mu-h\left(\sigma(\rho y+z)\right)/2, \ \ \sigma_{\Phi}(y,z):=\sqrt {h\left(\sigma(\rho y +z)\right)}, \\
&\mu_{\Psi}(y,z):=\frac{\kappa}{\sigma}\left(\theta-h\left(\sigma(\rho y +z)\right)\right)-\rho \mu_{\Phi}(y,z), \ \
\sigma_{\Psi}:=\sqrt{(1-\rho^2)} \sigma_{\Phi},
\end{eqnarray*}
and $\hat W:=\frac{\tilde W-\rho W}{\sqrt{1-\rho^2}}$ is a Brownian motion independent of $W$.

Next, we define
lattice based approximations for the process
$(\Phi,\Psi)$. Choose $\tilde\sigma\geq\overline{\sigma}$.
For any $n\in\mathbb N$ define the
stochastic processes $\Phi^{(n)}_t,\Psi^{(n)}_t$, $t\in [0,T]$ by
\begin{eqnarray*}
&\Phi^{(n)}_t:=\Phi_0+\tilde{\sigma}\sqrt{\frac{T}{n}}\sum_{i=1}^k\xi_i, \ \ \frac{kT}{n}\leq t<\frac{(k+1)T}{n}  \\
&\Psi^{(n)}_t:=\Psi_0+\tilde{\sigma}
\sqrt{\frac{T}{n}} \sum_{i=1}^k \hat\xi_i, \ \ \frac{kT}{n}\leq t<\frac{(k+1)T}{n}
\end{eqnarray*}
where $\xi_i,\hat\xi_i\in\{-1,0,1\}$.
Observe that the processes $\Phi^{(n)}-\Phi_0$, $\Psi^{(n)}-\Psi_0$
lie on the grid
$\tilde{\sigma}\sqrt\frac{T}{n}\{-n,1-n,...,n\}$.

Let $\mathcal F^{(n)}_t$, $t\leq T$ be the piece wise constant filtration generated by the processes $\Phi^{(n)},\Psi^{(n)}$.
Namely,
\begin{equation*}
\mathcal F^{(n)}_t:=\sigma\left\{\xi_1,...,\xi_k,\hat\xi_1,...,\hat\xi_k\right\}, \ \ k T/n\leq t<(k+1)T/n.
\end{equation*}
It remains to define the probability measure $\mathbb P_n$.
First since $W$ and $\hat W$ are independent Brownian motions we require that
for all $a,b\in\{-1,0,1\}$ and $k\geq 1$
\begin{equation*}
\mathbb P_n\left(\xi_k=a,\hat\xi_k=b|\mathcal F^{(n)}_{\frac{(k-1)T}{n}}\right):=\\
\mathbb P_n\left(\xi_k=a|\mathcal F^{(n)}_{\frac{(k-1)T}{n}}\right)
\mathbb P_n\left(\hat\xi_k=b|\mathcal F^{(n)}_{\frac{(k-1)T}{n}}\right).
\end{equation*}
In order to match the drift and the volatility,
we set,
\begin{eqnarray*}\label{7.5}
&\mathbb P_n\left(\xi_k=\pm 1| \mathcal
F^{(n)}_{\frac{(k-1)T}{n}}\right):=\frac{\sigma_{\Phi}^2\left(\Phi^{(n)}_{\frac{(k-1)T}{n}},\Psi^{(n)}_{\frac{(k-1)T}{n}}\right)}{2\tilde{\sigma}^2}\pm
\sqrt\frac{T}{n}\frac{\mu_{\Phi}\left(\Phi^{(n)}_{\frac{(k-1)T}{n}},\Psi^{(n)}_{\frac{(k-1)T}{n}}\right)}{2\tilde\sigma},\\
&\mathbb P_n\left(\xi_k=0| \mathcal F^{(n)}_{\frac{(k-1)T}{n}}\right):=
1-\frac{\sigma_{\Phi}^2\left(\Phi^{(n)}_{\frac{(k-1)T}{n}},\Psi^{(n)}_{\frac{(k-1)T}{n}}\right)}{\tilde{\sigma}^2},\label{7.6}
\end{eqnarray*}
and
\begin{eqnarray*}\label{7.7}
&\mathbb P_n\left(\hat \xi_k=\pm 1| \mathcal
F^{(n)}_{\frac{(k-1)T}{n}}\right):=\frac{\sigma_{\Psi}^2\left(\Phi^{(n)}_{\frac{(k-1)T}{n}},\Psi^{(n)}_{\frac{(k-1)T}{n}}\right)}{2\tilde{\sigma}^2}\pm
\sqrt\frac{T}{n}\frac{\mu_{\Psi}\left(\Phi^{(n)}_{\frac{(k-1)T}{n}},\Psi^{(n)}_{\frac{(k-1)T}{n}}\right)}{2\tilde\sigma},\\
&\mathbb P_n\left(\hat \xi_k=0| \mathcal F^{(n)}_{\frac{(k-1)T}{n}}\right):=
1-\frac{\sigma_{\Psi}^2\left(\Phi^{(n)}_{\frac{(k-1)T}{n}},\Psi^{(n)}_{\frac{(k-1)T}{n}}\right)}{\tilde{\sigma}^2}.\label{7.8}
\end{eqnarray*}
Observe that for sufficiently large $n$, the right hand side of the above equations all lie in the interval $[0,1]$.
\begin{prop}\label{lem.numerics}
For any $n\in\mathbb N$ (sufficiently large) consider the financial market given by
$S^{(n)}:=e^{\Phi^{(n)}}$ and the filtration $\mathcal F^{(n)}$ defined above. Then, the following holds true. \\
(I) We have the weak convergence
$S^{(n)}\Rightarrow S$ to the modified Heston model.\\
(II)
Assumption \ref{asm5.1} holds true.\\
\end{prop}
\begin{proof}
${}$\\
\textbf{Proof of (I).} Let us prove that
\begin{equation}\label{7.3}
(\Phi^{(n)},\Psi^{(n)})\Rightarrow (\Phi,\Psi).
\end{equation}
Clearly, (\ref{7.3}) implies that $S^{(n)}\Rightarrow S$.

From the definition of $\mathbb P_n$ we have
\begin{equation}\label{7.1000}
\mathbb E_{\mathbb P_n}\left(\Phi^{(n)}_{\frac{kT}{n}}-\Phi^{(n)}_{\frac{(k-1)T}{n}}
\big|\mathcal F^{(n)}_{\frac{(k-1)T}{n}}\right)=\frac{T}{n}\mu_{\Phi}\left(\Phi^{(n)}_{\frac{(k-1)T}{n}},\Psi^{(n)}_{\frac{(k-1)T}{n}}\right),
\end{equation}
\begin{equation}\label{7.2000}
\mathbb E_{\mathbb P_n}\left(\Psi^{(n)}_{\frac{kT}{n}}-\Psi^{(n)}_{\frac{(k-1)T}{n}}
\big|\mathcal F^{(n)}_{\frac{(k-1)T}{n}}\right)=\frac{T}{n}\mu_{\Psi}\left(\Phi^{(n)}_{\frac{(k-1)T}{n}},\Psi^{(n)}_{\frac{(k-1)T}{n}}\right),
\end{equation}
\begin{equation*}
\mathbb E_{\mathbb P_n}\left((\Phi^{(n)}_{\frac{kT}{n}}-\Phi^{(n)}_{\frac{(k-1)T}{n}})^2
\big|\mathcal F^{(n)}_{\frac{(k-1)T}{n}}\right)=
\frac{T}{n} \sigma^2_{\Phi}\left(\Phi^{(n)}_{\frac{(k-1)T}{n}},\Psi^{(n)}_{\frac{(k-1)T}{n}}\right),
\end{equation*}
\begin{equation*}
\mathbb E_{\mathbb P_n}\left((\Psi^{(n)}_{\frac{kT}{n}}-\Psi^{(n)}_{\frac{(k-1)T}{n}})^2
\big|\mathcal F^{(n)}_{\frac{(k-1)T}{n}}\right)=
\frac{T}{n} \sigma^2_{\Psi}\left(\Phi^{(n)}_{\frac{(k-1)T}{n}},\Psi^{(n)}_{\frac{(k-1)T}{n}}\right)
\end{equation*}
and
\begin{equation*}
\mathbb E_{\mathbb P_n}\left((\Phi^{(n)}_{\frac{kT}{n}}-\Phi^{(n)}_{\frac{(k-1)T}{n}})
(\Psi^{(n)}_{\frac{kT}{n}}-\Psi^{(n)}_{\frac{(k-1)T}{n}})\big|\mathcal F^{(n)}_{\frac{(k-1)T}{n}}\right)=O(n^{-2}).
\end{equation*}
Thus, (\ref{7.3}) follows from the
the martingale convergence result Theorem 7.4.1 in \cite{EK:86}.\\
\textbf{Proof of II.}
The statement follows from applying Example \ref{exm.bin} for
$m_n=n$, $\tau^{(n)}_i=(i-1)T/n$ and
$a_n=e^{\tilde{\sigma}\sqrt\frac{T}{n}}-1$.
\end{proof}
\subsection{Verification of Assumption \ref{asm4.1}}\label{sec:verification}
We start with some preparations.
Denote by $\mathcal D$ the set of all stochastic processes $\Upsilon=\{\Upsilon_t\}_{t=0}^T$ of the form
$\Upsilon=F(\Phi)$ where $F:\mathbb D([0,T];\mathbb R)\rightarrow \mathbb D([0,T];\mathbb R)$
is a bounded, continuous function
(we take the Skorokhod topology on the space
$\mathbb D([0,T];\mathbb R)$ and $F$ is a progressively measurable map. Namely,
for any $t\in [0,T]$ and $f^{(1)},f^{(2)}\in \mathbb D([0,T];\mathbb R)$,
$f^{(1)}_{[0,t]}=f^{(2)}_{[0,t]}$
implies that $F_t(f^{(1)})=F_t(f^{(2)})$.

Define the set
\begin{eqnarray*}
&\mathcal M^d(S):=\\
&\left\{\mathbb Q: \ \exists \Upsilon\in\mathcal D, \ \frac{d\mathbb Q}{d\mathbb P}|{\mathcal F^S_T}=e^{\int_{0}^{T}\frac{-\mu}{\sqrt {h(\nu_t)}} dW_t+\int_{0}^{T}\Upsilon_td\hat W_t-\int_{0}^{T} \frac{\mu^2}{2 h(\nu_t)} dt-
\int_{0}^{T} \frac{1}{2}\Upsilon^2_t dt} \right\}.
\end{eqnarray*}
From the Girsanov theorem it follows that $\mathcal M^d(S)\subset\mathcal M(S)$.
Moreover, since $\Phi=\ln S$ then
the usual filtration which is generated by $S$ equals to the usual filtration which is generated by $\Phi$.
Hence standard arguments yield that
$\mathcal M^d(S)\subset\mathcal M(S)$ is dense.

Choose an arbitrary $\Upsilon=F(\Phi)\in\mathcal D$ and
 denote
\begin{equation}\label{7.new}
Z_t:=e^{\int_{0}^{t}\frac{-\mu}{\sqrt {h(\nu_u)}} dW_u+\int_{0}^{t}\Upsilon_ud\hat W_u-\int_{0}^{t} \frac{\mu^2}{2 h(\nu_u)} du-
\int_{0}^{t} \frac{1}{2}\Upsilon^2_u du}, \ \ t\in [0,T].
\end{equation}
It is sufficient to prove that (recall Remark \ref{AE:4})
there exists a sequence of probability measures $\mathbb Q_n\in\mathcal M(S^{(n)})$, $n\in\mathbb N$, such that
for the processes $Z^{(n)}_t:=\frac{d\mathbb Q_n}{d\mathbb P_n}{|\mathcal F^{(n)}_t}$, $t\in [0,T]$,
we have the weak convergence
\begin{equation}\label{7.100}
(S^{(n)},Z^{(n)})\Rightarrow (S,Z).
\end{equation}
For any $n\in\mathbb N$ (sufficiently large) define the
 probability measure $\mathbb Q_n$ by the following relations
\begin{equation*}
\mathbb Q_n\left(\xi_k=a,\hat\xi_k=b|\mathcal F^{(n)}_{\frac{(k-1)T}{n}}\right):=\\
\mathbb Q_n\left(\xi_k=a|\mathcal F^{(n)}_{\frac{(k-1)T}{n}}\right)
\mathbb Q_n\left(\hat\xi_k=b|\mathcal F^{(n)}_{\frac{(k-1)T}{n}}\right),
\end{equation*}
\begin{equation}\label{7.11}
\begin{split}
\mathbb Q_n\left(\xi_k=\pm 1| \mathcal F^{(n)}_{\frac{(k-1)T}{n}}\right)&:=
\frac{\sigma_{\Phi}^2\left(\Phi^{(n)}_{\frac{(k-1)T}{n}},\Psi^{(n)}_{\frac{(k-1)T}{n}}\right)}
{\tilde{\sigma}^2\left(1+e^{\pm\tilde\sigma\sqrt\frac{T}{n}}\right)},\\
\mathbb Q_n\left(\xi_k=0| \mathcal F^{(n)}_{\frac{(k-1)T}{n}}\right)&:=
1-\frac{\sigma_{\Phi}^2\left(\Phi^{(n)}_{\frac{(k-1)T}{n}},\Psi^{(n)}_{\frac{(k-1)T}{n}}\right)}{\tilde{\sigma}^2},
\end{split}
\end{equation}
and
\[
\begin{split}
&\mathbb Q_n\left(\hat\xi_k=\pm 1| \mathcal F^{(n)}_{\frac{(k-1)T}{n}}\right):=
\frac{\sigma_{\Psi}^2\left(\Phi^{(n)}_{\frac{(k-1)T}{n}},\Psi^{(n)}_{\frac{(k-1)T}{n}}\right)}{2\tilde{\sigma}^2}\\
&\pm
\sqrt\frac{T}{n}\frac{F_{\frac{(k-1)T}{n}}(\Phi^{(n)})
\sigma_{\Psi}\left(\Phi^{(n)}_{\frac{(k-1)T}{n}},\Psi^{(n)}_{\frac{(k-1)T}{n}}\right)+
\mu_{\Psi}\left(\Phi^{(n)}_{\frac{(k-1)T}{n}},\Psi^{(n)}_{\frac{(k-1)T}{n}}\right)}{2\tilde\sigma},\nonumber\\
&\mathbb Q_n\left(\hat\xi_k=0| \mathcal F^{(n)}_{\frac{(k-1)T}{n}}\right):=
1-\frac{\sigma_{\Psi}^2\left(\Phi^{(n)}_{\frac{(k-1)T}{n}},\Psi^{(n)}_{\frac{(k-1)T}{n}}\right)}{\tilde{\sigma}^2}.\label{7.12}
\end{split}
\]
Observe that (\ref{7.11}) implies $\mathbb Q_n\in\mathcal M(S^{(n)})$.
\begin{lem}\label{lem5.2*}
We have the weak convergence
$$(\Phi^{(n)},\Psi^{(n)},Z^{(n)})\Rightarrow (\Phi,\Psi,Z).$$
\end{lem}
\begin{proof}
In order to prove the lemma it suffices to show that for any subsequence there exists a further subsequence
(still denoted by $n$) such that
\begin{equation}\label{5.300}
(\Phi^{(n)},\Psi^{(n)},Z^{(n)})\Rightarrow (\Phi,\Psi,Z).
\end{equation}
Fix $n\in\mathbb N$. By applying Taylor's expansion we obtain that there exist uniformly bounded (in $n$) processes
$E^{n,1}_k, E^{n,2}_k$, $k=0,1,...,n$ such that
$$
\frac{ \mathbb  Q_n\left(\xi_k| \mathcal F^{(n)}_{\frac{(k-1)T}{n}}\right)}
{\mathbb P_n\left(\xi_k| \mathcal F^{(n)}_{\frac{(k-1)T}{n}}\right)}=1-\tilde{\sigma}\xi_k\sqrt\frac{T}{n}
\left(\frac{1}{2}+\frac{\mu_{\Phi}\left(\Phi^{(n)}_{\frac{(k-1)T}{n}},\Psi^{(n)}_{\frac{(k-1)T}{n}}\right)}
{\sigma^2_{\Phi}\left(\Phi^{(n)}_{\frac{(k-1)T}{n}},\Psi^{(n)}_{\frac{(k-1)T}{n}}\right)}
\right)+\frac{E^{n,1}_{k}}{n}+o(1/n)$$
and
$$
\frac{\mathbb Q_n\left(\hat\xi_k| \mathcal F^{(n)}_{\frac{(k-1)T}{n}}\right)}
{\mathbb P_n\left(\hat\xi_k| \mathcal F^{(n)}_{\frac{(k-1)T}{n}}\right)}=1+\tilde{\sigma}\hat\xi_k\sqrt\frac{T}{n}
\frac{F_{\frac{(k-1)T}{n}}(\Phi^{(n)})}
{\sigma_{\Psi}\left(\Phi^{(n)}_{\frac{(k-1)T}{n}},\Psi^{(n)}_{\frac{(k-1)T}{n}}\right)}+\frac{E^{n,2}_k}{n}
+o(1/n).$$
We conclude that there exists a uniformly bounded (in $n$) process
$E^{(n)}_k$, $k=0,1...,n$ such that
\begin{equation}\label{5.main}
\begin{split}
\frac{Z^{(n)}_{\frac{k T}{n}}-Z^{(n)}_{\frac{(k-1)T}{n}}}
{Z^{(n)}_{\frac{(k-1)T}{n}}}&=\frac{\mathbb Q_n\left(\xi_k| \mathcal F^{(n)}_{\frac{(k-1)T}{n}}\right)}
{\mathbb P_n\left(\xi_k| \mathcal F^{(n)}_{\frac{(k-1)T}{n}}\right)}\frac{\mathbb Q_n\left(\hat\xi_k| \mathcal F^{(n)}_{\frac{(k-1)T}{n}}\right)}
{\mathbb P_n\left(\hat\xi_k| \mathcal F^{(n)}_{\frac{(k-1)T}{n}}\right)}-1\\
&=-(\Phi^{(n)}_{\frac{kT}{n}}-\Phi^{(n)}_{\frac{(k-1)T}{n}})\left(\frac{1}{2}+\frac{\mu_{\Phi}\left(\Phi^{(n)}_{\frac{(k-1)T}{n}},\Psi^{(n)}_{\frac{(k-1)T}{n}}\right)}
{\sigma^2_{\Phi}\left(\Phi^{(n)}_{\frac{(k-1)T}{n}},\Psi^{(n)}_{\frac{(k-1)T}{n}}\right)}
\right)\\
&+(\Psi^{(n)}_{\frac{k T}{n}}-\Psi^{(n)}_{\frac{(k-1)T}{n}})\frac{F_{\frac{(k-1)T}{n}}(\Phi^{(n)})}
{\sigma_{\Psi}\left(\Phi^{(n)}_{\frac{(k-1)T}{n}},\Psi^{(n)}_{\frac{(k-1)T}{n}}\right)}+\frac{E^{(n)}_k}{n}+o(1/n).
\end{split}
\end{equation}
In particular,
$\left(\frac{Z^{(n)}_{\frac{k T}{n}}-Z^{(n)}_{\frac{(k-1)T}{n}}}
{Z^{(n)}_{\frac{(k-1)T}{n}}}\right)^2$ is of order $O(1/n)$. Since
$Z^{(n)}$ is a martingale, then by taking conditional expectation we arrive to
$$\mathbb E_{\mathbb P_n}\left([Z^{(n)}_{\frac{k T}{n}}]^2|\mathcal F^{(n)}_{\frac{(k-1)T}{n}}\right)=[Z^{(n)}_{\frac{(k-1)T}{n}}]^2(1+O(1/n)).$$
By taking expectation we obtain
$$\mathbb E_{\mathbb P_n}\left([Z^{(n)}_{\frac{k T}{n}}]^2\right)=\mathbb E_{\mathbb P_n}\left([Z^{(n)}_{\frac{(k-1)T}{n}}]^2\right)(1+O(1/n)).$$
This together with the Doob--Kolmogorov inequality gives
\begin{equation}\label{uniform}
\sup_{n\in\mathbb N}\mathbb E_{\mathbb P_n}\left(\sup_{0\leq t\leq T}[Z^{(n)}_t]^2\right)\leq 4 \sup_{n\in\mathbb N}
\mathbb E_{\mathbb P_n}\left([Z^{(n)}_T]^2\right)<\infty.
\end{equation}
Next, define $\hat E^{(n)}_k:=\mathbb E_{\mathbb P_n}\left(E^{(n)}_k|\mathcal F^{(n)}_{\frac{(k-1)T}{n}}\right)$, $k=1,...,n$
and consider the martingale
$$\hat M^{(n)}_k:=\frac{1}{n}\sum_{i=1}^k ( E^{(n)}_i-\hat E^{(n)}_i), \ \ k=0,1,...,n.$$
Since $E^{(n)}$, $n\in\mathbb N$, are uniformly bounded then
\begin{eqnarray*}
&\mathbb E_{\mathbb P_n}\left(\max_{0\leq k\leq n}|\hat M^{(n)}_k|^2\right)\leq 4
\mathbb E_{\mathbb P_n}\left(|\hat M^{(n)}_n|^2\right)\\
&=\frac{4}{n^2}\sum_{i=1}^n \mathbb E_{\mathbb P_n}\left[\left( E^{(n)}_i-\hat E^{(n)}_i\right)^2\right]=O(1/n).
\end{eqnarray*}
Thus,
\begin{equation}\label{mar}
\max_{0\leq k\leq n}|\hat M^{(n)}_k|\rightarrow 0 \ \ \mbox{in} \ \mbox{probability}.
\end{equation}
Introduce the adapted (to $\mathcal F^{(n)}$) processes
\begin{eqnarray*}
&\Xi^{(n)}_t:=\int_{0}^t \hat E^{(n)}_{\floor*{nu/T}} du, \ \ t\in [0,T]\\
&M^{(n)}_t:= \hat M^{(n)}_{\floor*{nt/T}}, \ \ t\in [0,T]
\end{eqnarray*}
where $\floor*{\cdot}$ is the integer part of $\cdot$
and $\hat E^{(n)}_0:=E^{(n)}_0.$

Again, $E^{(n)}$, $n\in\mathbb N$, are uniformly bounded, and so
$\Xi^{(n)}$, $n\in\mathbb N$, is tight. From (\ref{7.3}) and (\ref{mar}) we conclude that the sequence
$(\Phi^{(n)},\Psi^{(n)},\Xi^{(n)},M^{(n)})$, $n\in\mathbb N$, is tight as well.
 Thus, from Prohorov's Theorem, (\ref{7.3}) and (\ref{mar}) it follows that for any subsequence there exists a further subsequence
such that
\begin{equation}\label{5.main1}
(\Phi^{(n)},\Psi^{(n)},\Xi^{(n)},M^{(n)})\Rightarrow (\Phi,\Psi,\Xi,0)
\end{equation}
 for some absolutely continuous process $\Xi=\{\Xi_t\}_{t=0}^T$.
From Theorems 4.3--4.4 in \cite{DP:92}, (\ref{5.main}), (\ref{5.main1})
and the equality
$\frac{E^{(n)}_k}{n}=\frac{\hat E^{(n)}_k}{n}+\hat M^{(n)}_k-\hat M^{(n)}_{k-1}$
we obtain that
$$(\Phi^{(n)},\Psi^{(n)},\Xi^{(n)},M^{(n)},Z^{(n)})\Rightarrow (\Phi,\Psi,\Xi,0,\hat Z)$$
where $\hat Z$ is the solution of the SDE
\begin{equation}\label{5.400}
\frac{d\hat Z_t}{\hat Z_t}=-\left(\frac{1}{2}+\frac{\mu_{\Phi}(\Phi_t,\Psi_t)}
{\sigma^2_{\Phi}(\Phi_t,\Psi_t)}\right)d\Phi_t+\frac{\Upsilon_t}{\sigma_{\Psi}(\Phi_t,\Psi_t)}d\Psi_t+\frac{d\Xi_t}{T}
\end{equation}
with the initial condition $\hat Z_0=1$.

Finally, (\ref{uniform}) implies that for any $t\in [0,T]$ the random variables
$\{Z^{(n)}_t\}_{n\in\mathbb N}$ are uniformly integrable. This together with the fact that
for any $n$, $Z^{(n)}$ is a martingale with respect to the filtration generated by
$\Phi^{(n)},\Psi^{(n)},\Xi^{(n)},M^{(n)},Z^{(n)}$ gives that $\hat Z$ is a martingale with respect to the filtration generated by
$\Phi,\Psi,\Xi,\hat Z$.
Moreover, from (\ref{7.1000})--(\ref{7.2000})
we get that
$\{\Phi_t-\int_{0}^t\mu_{\Phi}(\Phi_u,\Psi_u)du\}_{t=0}^T$
and $\{\Psi_t-\int_{0}^t\mu_{\Psi}(\Phi_u,\Psi_u)du\}_{t=0}^T$
are martingales with respect to the filtration generated by
$\Phi,\Psi,\Xi,\hat Z$.
In particular, from L\'evy's Theorem it follows that the stochastic processes $W$ and $\hat W$ which we redefine by
$$W_t:=\frac{\Phi_t-\int_{0}^t\mu_{\Phi}(\Phi_u,\Psi_u)du}{\sigma_{\Phi}(\Phi_t,\Psi_t)}, \  \hat
W_t:=\frac{\Psi_t-\int_{0}^t\mu_{\Psi}(\Phi_u,\Psi_u)du}{\sigma_{\Psi}(\Phi_t,\Psi_t)}$$
are (independent) Brownian motions
 with respect to the filtration generated by
$\Phi,\Psi,\Xi,\hat Z$.
We conclude that
the drift of the right hand side of (\ref{5.400}) is
equal to zero. Namely,
$$\frac{d\hat Z_t}{\hat Z_t}=-\left(\frac{1}{2}+\frac{\mu_{\Phi}(\Phi_t,\Psi_t)}{\sigma^2_{\Phi}(\Phi_t,\Psi_t)}\right)\sigma_{\Phi}(\Phi_t,\Psi_t)dW_t+
\Upsilon_td\hat W_t=\frac{dZ_t}{Z_t},$$
where the last equality follows from (\ref{7.new}).
Hence, $\hat Z=Z$ and (\ref{5.300}) follows.
\end{proof}
Clearly, Lemma \ref{lem5.2*} implies (\ref{7.100}). This gives us the following result.
\begin{prop}\label{cor1}
Consider the set-up of Proposition~\ref{lem.numerics}. Assumption \ref{asm4.1} holds true.
\end{prop}
We end this section by addressing condition (II) in Lemma \ref{AE:2}.
\begin{rem}
Consider the martingale measures $\mathbb Q_n\in\mathcal M(S^{(n)})$, $n\in\mathbb N$ which were defined
before Lemma \ref{lem5.2*} for $\Upsilon\equiv 0$.
Since $\mu^{\Phi},\sigma^{\Phi},\frac{1}{\sigma^ \Phi}$ are uniformly bounded, then
standard arguments yield that
for any $q>0$ (\ref{boundmeasure}) holds true.
\end{rem}

\section{Approximations of the Shortfall Risk in the Heston Model}\label{sec:6}

In this section we focus on shortfall risk minimization for European call options (which corresponds
to $U$ given by (\ref{5.utility})) in the Heston model.
We start with the following estimate.
\begin{lem}\label{error}
For an initial capital $x$ let $\hat R(x)$
be the shortfall risk in the Heston model and let
$R(x)$ be the shortfall risk in the model given by
(\ref{7.0-}). Then for any $m\in\mathbb N$
$$|\hat R(x)-R(x)|\leq O(\underline\sigma^{2\kappa\theta/\sigma^2-1})+O(1/\overline{\sigma}^{m}),$$
where the $O$ terms do not depend on $x$.
\end{lem}
\begin{proof}
Define the stopping time
$$\Theta_{\underline{\sigma},\overline{\sigma}}:=T\wedge\inf\{t: \sqrt{\hat\nu_t}\notin(\underline{\sigma},\overline{\sigma})\}.$$ Observe that on the event
$\Theta_{\underline{\sigma},\overline{\sigma}}=T$ the processes $\hat S$ and $S$ coincide. Hence,
\begin{equation}\label{1}
|\hat R(x)-R(x)|\leq \mathbb E_{\mathbb P}[(\hat S_T+S_T)\mathbb{I}_{\Theta_{\underline{\sigma},\overline{\sigma}}<T}]\leq 2e^{\mu T} \mathbb E_{\mathbb
P}[e^{-\mu\Theta_{\underline{\sigma},\overline{\sigma}}}
\hat S_{\theta_{\underline{\sigma},\overline{\sigma}}} \mathbb{I}_{\Theta_{\underline{\sigma},\overline{\sigma}}<T}]
\end{equation}
where the last inequality is due to the fact that
the processes $e^{-\mu t}\hat S_t, e^{-\mu t}S_t$, $t\in [0,T]$ are martingales.

Introduce the probability measure $\mathbf P$ by
$\frac{d\mathbf P}{d\mathbb P}|{\mathcal F^S_T}:=\frac{e^{-\mu\Theta_{\underline{\sigma},\overline{\sigma}}}\hat S_{\Theta_{\underline{\sigma},\overline{\sigma}}}}{S_0}.$
Then by Girsanov theorem
the process
$\mathbf W_t:=\tilde W_t-\rho\int_{0}^{t\wedge\Theta_{\underline{\sigma},\overline{\sigma}}} \sqrt{\hat\nu_u}$, $t\in [0,T]$, is a Brownian motion with
respect to $\mathbf P$.
Let $\{\alpha_t\}_{t=0}^T$ be the unique strong solution of the SDE
$$d\alpha_t=\left(\kappa(\theta-{\alpha_t})+\sigma\rho\alpha_t\right)dt+\sigma \sqrt{ \alpha_t}d\mathbf  W_t, \ \ \alpha_0=\hat \nu_0.$$
Observe that
\begin{equation}\label{box}
\alpha_{[0,\Theta_{\underline{\sigma},\overline{\sigma}}]}=\hat\nu_{[0,\Theta_{\underline{\sigma},\overline{\sigma}}]}.
\end{equation}
Clearly, for any $m\in\mathbb N$ we have
$$\mathbb E_{\mathbf P}\left(\sup_{0\leq t\leq T} [\sqrt\alpha_t]^{m}\right)<\infty.$$ Thus, from the Markov inequality
we get
\begin{equation}\label{2}
\mathbf P\left(\sup_{0\leq t\leq T}\sqrt\alpha_t\geq\overline{\sigma}\right)=O(1/\overline{\sigma}^{m}), \ \ \forall m\in\mathbb N.
\end{equation}
Moreover, from Theorem 2 in \cite{JY:03} it follows that
\begin{equation}\label{3}
\mathbf P\left(\inf_{0\leq t\leq T}\sqrt\alpha_t\leq\underline{\sigma}\right)=O(\underline\sigma^{2\kappa\theta/\sigma^2-1}).
\end{equation}
 By combining (\ref{1})--(\ref{3}) we conclude that
\[
\begin{split}
|\hat R(x)-R(x)&|\leq 2S_0e^{\mu T}\left(\mathbf P\left(\inf_{0\leq t\leq T}\sqrt\alpha_t\leq\underline{\sigma}\right)+
\mathbf P\left(\sup_{0\leq t\leq T}\sqrt\alpha_t\geq\overline{\sigma}\right)\right)\\
 &\leq O(\underline\sigma^{2\kappa\theta/\sigma^2-1})+O(1/\overline{\sigma}^{m})
\end{split}
\]
 as required.
\end{proof}

Next, we focus
on approximating the shortfall risk in the model given by (\ref{7.0-}).
In order to apply
Theorem \ref{main}
we need to verify the required Assumptions.
Observe that
Assumption \ref{asm2.1}, Assumption \ref{asm2.2}(ii) ($U^{+}\equiv 0$) and Assumption \ref{asm4.2} trivially hold true.
Moreover, from Remark \ref{AE:new}
we obtain
Assumption \ref{asm2.1+}.
Since the drift and the volatility are uniformly bounded we
get that the random variables $\{S^{(n)}_T\}_{n\in\mathbb N}$
are uniformly integrable, which gives
Assumption \ref{asm2.2}(i).
In view of
Propositions
\ref{lem.numerics},\ref{cor1} we conclude
that our Assumptions
are satisfied and so Theorem \ref{main} holds true.

Thus, fix $n\in\mathbb N$ and recall the discrete models introduced in Section \ref{sec:dis}.
The stock price process $S^{(n)}$ is piece wise constant and so the investor trades only at the jump times
$\frac{kT}{n}$, $k=0,1...,n.$
Notice that
$\left\{\sum_{m=1}^k \xi_m,\sum_{m=1}^k \hat \xi_m\right\}_{k=0}^n$ is a lattice valued Markov chain (with respect to $\mathbb P_n$).
Hence, we introduce the functions
$J^{(n)}_k(i,j,\lambda)$, $k=0,1...,n$ such that $J^{(n)}_k(i,j,\lambda)$
measures the shortfall risk at time $kT/n$ given that
$\sum_{m=1}^k \xi_m=i$,
$\sum_{m=1}^k \hat\xi_m=j$, and
$\lambda$ is the ratio of the portfolio value and the stock price.
The stock price is recovered by
$$S^{(n)}_{\frac{kT}{n}}=S_0 e^{\tilde\sigma\sqrt\frac{T}{n}\sum_{m=1}^k\xi_m}=S_0 e^{i\tilde\sigma\sqrt\frac{T}{n}}.$$
Clearly, if $\lambda\geq 1$, then the shortfall risk is zero because we can buy the stock and hold it until maturity. Namely,
$J^{(n)}_k(i,j,\lambda)=0$ for $\lambda\geq 1$. Hence, we assume that
$\lambda\in [0,1]$.

Next, we describe the dynamic programming principle to solve the discrete control-problem.
At time $kT/n$ the investor decides about his investment policy.
Assume that the investor portfolio value is $\lambda S^{(n)}_{\frac{kT}{n}}$.
We have a trinomial model with growth rates
$\left\{e^{-\tilde{\sigma}\sqrt{\frac{T}{n}}},1,e^{\tilde{\sigma}\sqrt{\frac{T}{n}}}\right\}$.
From the binomial representation theorem we easily deduce that the set of replicable portfolios at time $(k+1)T/n$ are of the form
$\Lambda(\xi_{k+1}) S^{(n)}_{\frac{(k+1)T}{n}}$ where
$\Lambda:\{-1,0,1\}\rightarrow\mathbb R$ satisfies $\Lambda(0)=\lambda$ and
$$\frac{\Lambda(-1)+\Lambda(1)e^{\tilde{\sigma}\sqrt{\frac{T}{n}}}}{1+e^{\tilde{\sigma}\sqrt{\frac{T}{n}}}}=\lambda.$$
Thus, if $\Lambda(-1)$ is known then we set
\begin{equation}\label{6.0}
\Lambda(1):=1\wedge\left(\lambda(1+e^{-\tilde{\sigma}\sqrt{\frac{T}{n}}})-\Lambda(-1)e^{-\tilde{\sigma}\sqrt{\frac{T}{n}}}\right).
\end{equation}
We take a truncation in order to have $\Lambda(1)\in [0,1]$.
In view of our admissibility condition, we denote by $\mathcal A(\lambda)$ the set of all $\Lambda(-1)\in [0,1]$
for which the right hand side of
(\ref{6.0}) is non-negative.

We arrive to the following recursive relations.
Define
$$J^{(n)}_{k}(i,j,\lambda):\{-k,1-k,...,k\}\times\{-k,1-k,...,k\} \times [0,1]\rightarrow\mathbb R_{+}, \ \ k=0,1,...,n$$
by
$$J^{(n)}_{n}(i,j,\lambda):=U\left(\lambda S_0\exp\left(i\tilde{\sigma}\sqrt\frac{T}{n}\right),S_0\exp\left(i\tilde{\sigma}\sqrt\frac{T}{n}\right)\right),$$
and for $k<n$,
\begin{equation}\label{6.1--}
\begin{split}
&J^{(n)}_k(i,j,\lambda):=\\
&\sup_{\Lambda(-1)\in\mathcal A(\lambda)}\mathbb E_{\mathbb P_n}\bigg(J^{(n)}_{k+1}\bigg(i+\xi_{m+1},j+\hat\xi_{m+1},\Lambda(\xi_{m+1})\bigg)\bigg|\sum_{m=1}^k\xi_m=i, \
\sum_{m=1}^k\hat\xi_m=j\bigg)
\end{split}
\end{equation}
where
$\Lambda(0)=\lambda$ and
$\Lambda(1)$ is given by (\ref{6.0}).
For $k=0$ we
have
$J^{(n)}_0(x/S_0)=u_n(x).$

Observe that the functions $J^{(n)}_k(i,j,\lambda)$ are piece wise linear and continuous in $\lambda$, and so
they can be represented by an array which consists of the slope values and the slope jump points.
This together with the condition $J^{(n)}_k(i,j,1)=0$ is sufficient to recover the function.
Of course the array will depend on time $kT/n$ and the states $i,j$. Thus, theoretically, the dynamic programming given by (\ref{6.1--}) can be implemented using a computer.
However, from practical point of view
the number of the slope points of the function $J^{(n)}_k$ grows exponentially (in $n-k$),
and so for large $n$ it cannot be implemented.
Hence, we need to introduce a grid structure for the portfolio value as well.

Thus, choose $M\in\mathbb N$ and consider the grid
\begin{equation}\label{eq:cgrid}
GR:=\left\{0,\frac{1}{M},\frac{2}{M},...,1\right\}.
\end{equation}

For a given $\Lambda(-1)\in GR$ we define two grid values for $\Lambda(1)$.
The first value is
\begin{equation}\label{6.1}
\Lambda^{-}(1):=1\wedge\frac{\floor*{\left(\lambda(1+e^{-\tilde{\sigma}\sqrt{\frac{T}{n}}})-\Lambda(-1)e^{-\tilde{\sigma}\sqrt{\frac{T}{n}}}\right)M}}{M}
\end{equation}
where, recall that $\floor*{\cdot}$ is the integer part of $\cdot$.
The second value is
\begin{equation}\label{6.2}
\Lambda^{+}(1):=1\wedge\frac{\ceil*{\left(\lambda(1+e^{-\tilde{\sigma}\sqrt{\frac{T}{n}}})-\Lambda(-1)e^{-\tilde{\sigma}\sqrt{\frac{T}{n}}}\right)M}+1}{M}
\end{equation}
where $\ceil*{\cdot}=\min\{n\in\mathbb Z: n\geq \cdot\}$.
Define two value functions
\begin{equation}\label{eq:J-pm-func}
J^{(n)}_{k}(\pm,i,j,\lambda):\{-k,1-k,...,k\}\times\{-k,1-k,...,k\} \times GR\rightarrow\mathbb R_{+}, \ \ k=0,1,...,n
\end{equation}
as following.
The terminal condition is
$$J^{(n)}_{n}(\pm, i,j,\lambda):=U\left(\lambda S_0\exp\left(i\tilde{\sigma}\sqrt\frac{T}{n}\right),S_0\exp\left(i\tilde{\sigma}\sqrt\frac{T}{n}\right)\right).$$
For $k<n$,
\[
\begin{split}
&J^{(n)}_k(\pm,i,j,\lambda)\\
&:=\max_{\Lambda(-1)\in\mathcal A(\lambda)\bigcap GR}\mathbb E_{\mathbb P_n}\bigg(J^{(n)}_{k+1}\bigg(\pm,i+\xi_{m+1},j+\hat\xi_{m+1},\\
&\Lambda^{\pm}(\xi_{m+1})\bigg)\bigg|\sum_{m=1}^k\xi_m=i, \ \sum_{m=1}^k\hat\xi_m=j\bigg)\nonumber
\end{split}
\]
where
$\Lambda^{\pm}(-1)=\Lambda(-1)$, $\Lambda^{\pm}(0)=\lambda$ and
$\Lambda^{\pm}(1)$ are given by (\ref{6.1})--(\ref{6.2}).

For $k=0$ we obtain two values
$J^{(n)}_0(+,x/S_0)$ and $J^{(n)}_0(-,x/S_0)$. Observe that the complexity of the above dynamic programming is polynomial in $M,n$.
For the exact value $u_n(x)=J^{(n)}_0(x/S_0)$ we have the following simple lemma.
\begin{lem}\label{lem6.1}
Assume that $\frac{x}{S_0}\in GR$. Then
$$J_n(x/S_0)\in [J^{(n)}_0(-,x/S_0),J^{(n)}_0(+,x/S_0)].$$
\end{lem}
\begin{proof}
The inequality $J^{(n)}_0(-,x/S_0)\leq J^{(n)}_0(x/S_0)$ is obvious. Let us prove that
$J^{(n)}_0(x/S_0)\leq J^{(n)}_0(+,x/S_0)$.
Choose $\lambda\in GR$ and
$\tilde \Lambda(-1),\tilde \Lambda(1)\in [0,1]$ which satisfy (\ref{6.0}). Define
$\Lambda(-1):=1\wedge\frac{\ceil*{\tilde \Lambda(-1)M}}{M}$ and let
$\Lambda^{+}(1)$ be given by (\ref{6.2}).
Then it is straightforward to check that $\Lambda(-1)\geq \tilde \Lambda(-1)$ and
$\Lambda^{+}(1)\geq\tilde\Lambda(1).
$
Hence, by applying backward induction (on $k$) and the fact that
$J^{(n)}_k(i,j,\lambda)$ is non-decreasing in $\lambda$ we get
that for any $k$, $J^{(n)}_k(\cdot)\leq  J^{(n)}_k(+,\cdot)$
where we take the restriction of $J^{(n)}_k(\cdot)$ to $\{-k,1-k,...,k\}\times\{-k,1-k,...,k\} \times GR$.
For $k=0$, we obtain $J^{(n)}_0(x/S_0)\leq J^{(n)}_0(+,x/S_0)$ as required.
\end{proof}
\begin{rem}\label{rem:grdM}
By using the fact that $U$ is Lipschitz continuous in the first variable, it can be shown that
the difference $J^{(n)}_0(+,x/S_0)-J^{(n)}_0(-,x/S_0)$ is of order $O(n/M)$.
In practice this difference goes to zero much faster (in $M$). As we will see in the
following numerical results, already for $M$ ``close" to $n$ the difference
$J^{(n)}_0(+,x/S_0)-J^{(n)}_0(-,x/S_0)$ becomes very small.
\end{rem}
\subsection{Numerical Results}

In this section we implement numerically the above described procedure.
In Table \ref{Table1new1} and in the corresponding Figure \ref{figure1} we compute the functions defined in \eqref{eq:J-pm-func}.
To serve as a reference we also evaluate the function
$
\overline{u}(x)=-\mathbb E_P\left[((S_T-K)^+-x)^+\right],
$
a lower bound, which corresponds to the value of spending no extra effort in reducing the shortfall.

\begin{table}[H]
\begin{center}
\begin{tabular}{c c c c c}
&$J_0^{(n)}(-,x/S_0)$&$J_0^{(n)}(+,x/S_0)$&$\overline{u}_0^{(n)}(x)$\\
\hline
x=0&-24.5421&-24.0371&-24.6095\\
\hline\\
x=5&-18.4702&-17.7050&-21.4086\\
\hline\\
x=10&-12.3159&-11.6165&-18.2077\\
\hline\\
x=15&-7.0529&-6.3398&-16.3018\\
\hline\\
x=20&-2.7913&-2.2453&-14.3959\\
\hline\\
x=25&-0.6802&-0.4201&-12.4901\\
\hline\\
x=30&-0.0825&-0.0274&-10.5842\\
\hline\\
x=35&-0.0043&-0.0004&-8.6783\\
\hline\\
x=40&0&0&-7.1540\\
\hline\\
x=45&0&0&-6.4423\\
\hline\\
x=50&0&0&-5.7306\\
\hline\\
x=55&0&0&-5.0190\\
\hline\\
x=60&0&0&-4.3073\\
\hline\\
x=70&0&0&-2.8840\\
\hline\\
x=80&0&0&-2.0045\\
\hline\\
x=90&0&0&-1.6418\\
\hline\\
x=100&0&0&-1.2792\\
\hline\\
\end{tabular}
\end{center}
\caption{Shortfall risk minimization for call options. Parameters used in computation:
$K=90,\overline{\sigma}=1,\tilde{\sigma}=5,\underline{\sigma}=0.0001;\sigma=0.39,\rho=-0.64,\kappa=1.15,\theta=0.348,\mu=0.05,S_0=100,T=1,\nu_0=0.09,n=400,M=400$. }
\label{Table1new1}
\end{table}

\begin{figure}[H]
\begin{center}
\includegraphics[scale=.7]{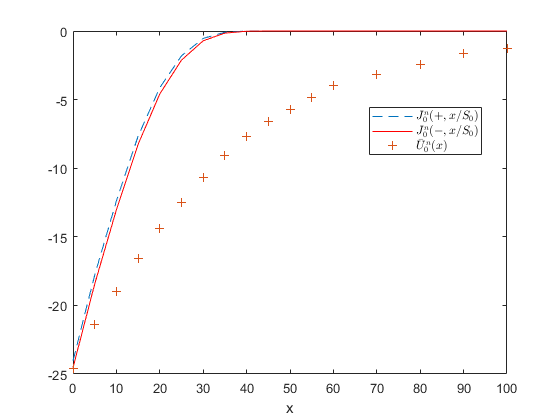}
\end{center}
\caption{ Plot of the values reported in Table~\ref{Table1new1}.}
\label{figure1}
\end{figure}

In the next table we analyze the sensitivity of the problem to $\overline{\sigma}$. The smaller this parameter, the faster the algorithm takes. Although, Lemma~\ref{error}
indicates an error bound for large $\overline{\sigma}$ (which was obtained by an application of Markov's inequality), we observe that we can in practice take
$\overline{\sigma}=1$ for our parameters.

\begin{table}[H]
\begin{center}
\begin{tabular}{c c c c c c c}
&$\overline{\sigma}=0.4(0.1757)$&$\overline{\sigma}=0.6(0.8085)$&$\overline{\sigma}=0.8(0.9939)$&$\overline{\sigma}=1$&$\overline{\sigma}=2$\\
\hline
x=0&-15.3139&-23.1077&-22.0861&-24.5421&-24.5421\\
\hline\\
x=10&-4.1129&-9.6884&-10.9334&-12.3159&-12.3159\\
\hline\\
x=20&-0.1435&-4.5287&-1.9145&-2.7913&-2.7913\\
\hline\\
\end{tabular}
\end{center}
\caption{Variation with respect to $\overline{\sigma}$. Parameters are the same as in Table~\ref{Table1new1}. The values in the parentheses represent
$\mathbb P(\Theta_{\underline{\sigma},\overline{\sigma}}<T)$ rounded to 4 decimals points. We did not indicate these values when this probability is extremely close to 1. }
\label{Table1new2}
\end{table}

In Table~\ref{TableJne} we analyze the sensitivity of solution to the grid size of the control  variable defined in  \eqref{eq:cgrid}. We observe, as stated in
Remark~\ref{rem:grdM}, that the we can actually take $M= kn$, where $k<1$. In this table, we determine the range of $k$ we can choose. We observe that choosing $n$ larger
leads to more error reduction than choosing $k$ larger. We have also checked this for values of $k>1$.

\begin{table}[H]
\begin{center}
\begin{tabular}{c c c c c}
&M=n/4&M=n/2&M=n\\
\hline
n=50&-9.2138&-6.6971&-6.6586\\
\hline
n=100&-5.4667&-5.4282&-5.4238\\
\hline
n=200&-3.7184&-3.6541&-3.6448\\
\hline
n=400&-2.9834&-2.8392&-2.7913\\
\hline
n=800&-2.6675&-2.5299&-2.4833\\
\hline
\end{tabular}
\end{center}
\caption{Variation with respect to $M$. $x=20$. Other parameters are the same as in Table~\ref{Table1new1}.}
\label{TableJne}
\end{table}

Table~\ref{TableJneerror} and the corresponding Figure~\ref{TableJneerror} demonstrate the convergence with respect to $n$. We observe that the convergence rate is a power of
$n$. We leave the rigorous demonstration of this result for future work.

\begin{table}[H]
\begin{center}
\begin{tabular}{c c c}
&M=n/4& $\frac{J_0^n(-,x/S_0)-J_0^{n/2}(-,x/S_0)}{|J_0^{n/2}(-,x/S_0)|}$\\
\hline
n=50&-9.2138&--\\
\hline
n=100&-5.4667&0.4067\\
\hline
n=200&-3.7184&0.3198\\
\hline
n=400&-2.9834&0.1977\\
\hline
n=800&-2.6675&0.1059\\
\hline
n=1600&-2.6171&0.0189\\
\hline
\end{tabular}
\end{center}
\caption{$x=20$. Other parameters are the same as in Table~\ref{Table1new1}.}
\label{TableJneerror}
\end{table}

\begin{figure}[H]
\begin{center}
\includegraphics[scale=.7]{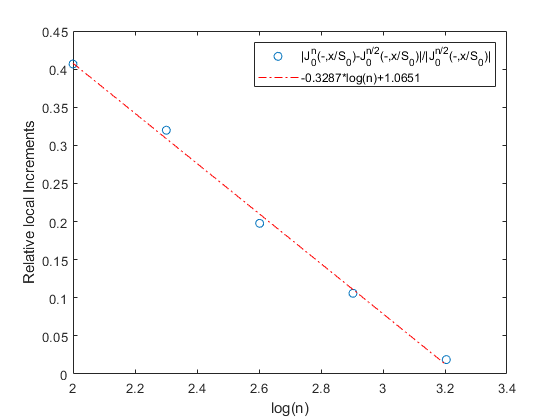}
\end{center}
\caption{Plot of the values in Table~\ref{TableJneerror}.}
\label{figure2}
\end{figure}

\bibliographystyle{plainnat}
\bibliography{finance}

\end{document}